\documentclass[12pt]{article}

\usepackage{graphics,graphicx,fullpage,natbib,multirow}
\usepackage{amsmath,amssymb,verbatim,epsfig}
\usepackage[dvipsnames,usenames]{color}
\usepackage{caption,subcaption,booktabs}

\newtheorem{proposition}{Proposition}

\newtheorem{defin}{\bf Definition}

\newenvironment{proof}{\noindent{\bf Proof}}{$\diamond$}

\def\ga{\mbox{Ga}}

\def\be{\mbox{Be}}
\def\ber{\mbox{Ber}}
\def\bin{\mbox{Bin}}
\def\mul{\mbox{Mul}}
\def\nbi{\mbox{NB}}

\def\no{\mbox{N}}
\def\po{\mbox{Po}}
\def\pg{\mbox{Pg}}

\def\un{\mbox{Un}}

\def\E{\mbox{E}}
\def\V{\mbox{Var}}
\def\Cov{\mbox{Cov}}
\def\Cor{\mbox{Corr}}

\def\rest{\mbox{rest}}

\def\bw{{\bf w}}
\def\bx{{\bf x}}
\def\by{{\bf y}}

\def\bW{{\bf W}}
\def\bX{{\bf X}}
\def\bY{{\bf Y}}

\def\simind{\stackrel{\mbox{\scriptsize{ind}}}{\sim}}
\def\simiid{\stackrel{\mbox{\scriptsize{iid}}}{\sim}}
\newcommand{\balpha}{\boldsymbol{\alpha}}

\newcommand{\btheta}{\boldsymbol{\theta}}

\newcommand{\NN}{\mathbb{N}}

\newcommand{\SC}{\mathbb{S}}
\newcommand{\TT}{\mathbb{T}}

\newcommand{\XX}{\mathcal{X}}

\begin{document}

\baselineskip=24pt

\title{\bf Dependence on a collection of Poisson random variables}
\author{{\sc Luis E. Nieto-Barajas} \\[2mm]
{\sl Department of Statistics, ITAM, Mexico} \\[2mm]
{\small {\tt lnieto@itam.mx}} \\}
\date{}
\maketitle

\begin{abstract}
We propose two novel ways of introducing dependence among Poisson counts through the use of latent variables in a three levels hierarchical model. Marginal distributions of the random variables of interest are Poisson with strict stationarity as special case. Order--$p$ dependence is described in detail for a temporal sequence of random variables, however spatial or spatio-temporal dependencies are also possible. A full Bayesian inference of the models is described and performance of the models is illustrated with a numerical analysis of maternal mortality in Mexico. Extensions to cope with overdispersion are also discussed.
\end{abstract}

\vspace{0.2in} \noindent {\sl Keywords}: Autoregressive process, integer-valued time series, latent variables, moving average process, stationary process.

\section{Introduction}
\label{sec:intro}

Time series models are mainly discrete time stationary processes. The support of the random variables involved is usually continuous and unbounded \citep[e.g.][]{box&jenkins:70}. The study of discrete time stationary processes with discrete marginal distributions is less common, however there have been some proposals \citep[e.g.][]{mckenzie:85}. 

In this article we define discrete time stochastic processes with Poisson marginal distributions. Construction of our proposal is based on the use of latent variables, through hierarchical models, which allows us to define different orders of dependence in space and time. Overdispersion is possible to handle by a straightforward generalisation, defining negative binomial marginal distributions. 

Before we proceed we introduce some notation: $\ber(\alpha)$ denotes a Bernoulli distribution with success probability $\alpha$; $\bin(n,\alpha)$ denotes a binomial distribution with $n$ Bernoulli trials and success probability $\alpha$; $\po(\mu)$ denotes a Poisson distribution with mean (rate) $\mu$; $\no(\mu,\tau)$ denotes a normal distribution with mean $\mu$ and precision $\tau$; $\mul(n,\balpha)$ denotes a multinomial distribution with $n$ number of trials and vector of probabilities $\balpha$. In general, we will add an argument upfront to denote the corresponding density, e.g. $\ber(x\mid\alpha)$ denotes a Bernoulli density evaluated at $x$. 

The main building block of our proposal is: 
\begin{align}
\nonumber
X\sim\po(\mu)\mbox{ and }&Y\mid X=x\sim\bin(x,\alpha)\\ 
\label{eq:main} 
&\iff \\ 
\nonumber
Y\sim\po(\mu\alpha)\mbox{ and }&X-y\mid Y=y\sim\po(\mu(1-\alpha)). 
\end{align}
This result \eqref{eq:main} is straightforward to prove by using probability calculus. 

One of the first proposals in the literature is the integer-valued first order autoregressive process, INAR(1), which for a process $\{X_t\}$ is defined as \citep{mckenzie:85,alosh&alzaid:87}
\begin{equation}
\label{eq:inar1}
X_t=\alpha\circ X_{t-1}+\epsilon_t,
\end{equation}
where ``$\circ$'' denotes the binomial thinning operator defined as $\alpha\circ X=\sum_{j=1}^X B_j$ with $B_j\simiid\ber(\alpha)$. In other words $\alpha\circ X\mid X=x\sim\bin(x,\alpha)$. If we denote $Y_{t}\equiv\alpha\circ X_{t-1}$ in \eqref{eq:inar1}, then $Y_{t}\mid X_{t-1}=x_{t-1}\sim\bin(x_{t-1},\alpha)$. Moreover, if the innovations  are Poisson distributed, $\epsilon_{t}\sim\po(\mu(1-\alpha))$, then $X_t-y_t\mid Y_t=y_t\sim\po(\mu(1-\alpha))$. Thus if $X_{t-1}\sim\po(\mu)$, result \eqref{eq:main} implies that marginally $X_t\sim\po(\mu)$. In summary, if $X_0\sim\po(\mu)$ and $\{\epsilon_t\}$ is a sequence of i.i.d. $\po(\mu(1-\alpha))$, then $\{X_t\}$ is a stationary process with $\po(\mu)$ marginal distributions. The autocorrelation function of \eqref{eq:inar1} can be obtained analytically and has the form $\Cor(X_t,X_{t+s})=\rho(s)=\alpha^s$ for $s\geq 0$ \citep{mckenzie:85}. 

Later, \cite{mckenzie:88} generalized the  INAR(1) process to the ARMA type. For instance, the Poisson MA$(q)$ process is defined as 
\begin{equation}
\label{eq:maq}
X_t=Z_t+\beta_1\circ Z_{t-1}+\cdots+\beta_q\circ Z_{t-q},
\end{equation}
where $\beta_i\in(0,1)$ for all $i$, and $Z_t\simind\po(\mu/\beta)$ with $\beta=\sum_{i=0}^q\beta_i$ and $\beta_0=1$. Denoting by $Y_{i}=\beta_i\circ Z_{t-i}$ then $Y_i\mid Z_{t-i}=z_{t-i}\sim\bin(z_{t-i},\beta_i)$, and from \eqref{eq:main}, $Y_i\simind\po(\mu\beta_i/\beta)$ marginally. Now, using the additive property of independent Poisson random variables, it becomes that $X_t\sim\po(\mu)$. The autocorrelation function of \eqref{eq:maq} is given by $\rho(s)=\sum_{i=0}^{q-s}\beta_i\beta_{i+s}/\sum_{i=0}^q\beta_i$ for $s\leq q$, and zero otherwise. 

Another generalization of INAR(1) process is that of \cite{alzaid&alosh:90}, who proposed the INAR(p) process as follows 
\begin{equation}
\label{eq:inarp}
X_t=\sum_{i=1}^p \alpha_i\circ X_{t-i}+\epsilon_t,
\end{equation}
where $\alpha_i>0$ for all $i$ with $\sum_{i=1}^p\alpha_i<1$, and the conditional distribution of the vector $(\alpha_1\circ X_t,\alpha_2\circ X_t,\ldots,\alpha_p\circ X_t)\mid X_t=x_t\sim\mul(x_t,\balpha)$ with $\balpha=(\alpha_1,\alpha_2,\ldots,\alpha_p)$. Even if the distribution for the innovations $\epsilon_t$ in \eqref{eq:inarp} is Poisson, the marginal distribution of $X_t$ is not Poisson. 

More recent approaches considered generalized linear models, where $X_t$ is assumed Poisson distributed with mean $\mu_t$. The intensity $\mu_t$ (or $\log\mu_t$) is further regressed on lagged values $X_{t-i}$ (or $\log X_{t-i}$) and $\mu_{t-i}$ (or $\log\mu_{t-i}$), for positive integer $i$, producing what are called integer GARCH models \citep{fokianos&kedem:04,fokianos&al:09}. \cite{chen&lee:16} also work with generalized Poisson autoregressive models but with a switching mechanism and with zero-inflation, and \cite{chen&lee:17} further propose a causality test for the same type of models. As proved by \cite{ferland&al:06} these models are second order stationary under some conditions, but the marginal distribution is not Poisson. This is not a problem, but sometimes a desirable feature for strict stationarity in time series analysis. A summary of the state of the art integer-valued models can be found in \cite{davis&al:16}.

The description of the rest of the paper is as follows: In Section \ref{sec:model} we describe the construction of two dependent Poisson sequences in time and characterise its marginal distribution and correlation induced. Bayesian inference of model parameters is described in Section \ref{sec:inference}. Section \ref{sec:numerical} reports a numerical study of integer-valued time series of maternal mortality in Mexico. Section \ref{sec:extensions} presents some extensions to more general dependencies, as seasonal, periodic and spatial. We conclude with some remarks in Section \ref{sec:conclude}, where we also discuss the generalisation to negative binomial marginal distributions.

\section{Temporal dependence}
\label{sec:model}

Let $\{X_t\}_{t\in\NN}$ be a stochastic process indexed by $t\in\NN$, where $\NN=\{1,2,\ldots\}$ denotes the set of natural numbers. For each $t$ we require a set of two latent variables, say $(Y_t,W_t)$, and define a three level hierarchical model to induce a temporal dependence of order $p\geq 0$. Additionally, $Y_t$ and $W_t$ will exist for $t\in\TT$ with $\TT=\{1-p,-p,-p+1\ldots\}$. Let $\bY=\{Y_t\}_{t\in\TT}$ and $\bW=\{W_t\}_{t\in\TT}$. We propose two ways of defining dependence among the $X_t$'s by either, linking the variables of the second level with those of the third level across times (type A), or linking the variables of the first level to those of the second level across times (type B). Figure \ref{fig:graph} illustrates these two types, where the dependence shown is of order $p=1$. 

\begin{figure}[ht]
\setlength{\unitlength}{0.8cm}
\begin{center}
$(A)$ \hspace{7.5cm} $(B)$
\begin{picture}(20,8)
\put(1.3,7.0){$W_0$} 
\put(2.8,7.0){$W_1$} 
\put(4.3,7.0){$W_2$} 
\put(5.8,7.0){$W_3$} 
\put(7.3,7.0){$W_4$} 
\put(1.5,6.7){\vector(0,-1){2}}
\put(3.0,6.7){\vector(0,-1){2}}
\put(4.5,6.7){\vector(0,-1){2}}
\put(6.0,6.7){\vector(0,-1){2}}
\put(7.5,6.7){\vector(0,-1){2}}
\put(1.3,4.0){$Y_0$} 
\put(2.8,4.0){$Y_1$} 
\put(4.3,4.0){$Y_2$} 
\put(5.8,4.0){$Y_3$} 
\put(7.3,4.0){$Y_4$} 
\put(3.0,3.7){\vector(0,-1){2}}
\put(4.5,3.7){\vector(0,-1){2}}
\put(6.0,3.7){\vector(0,-1){2}}
\put(7.5,3.7){\vector(0,-1){2}}
\put(2.8,1.0){$X_1$} 
\put(4.3,1.0){$X_2$} 
\put(5.8,1.0){$X_3$} 
\put(7.3,1.0){$X_4$} 
\put(1.5,3.7){\vector(2,-3){1.35}}
\put(3.0,3.7){\vector(2,-3){1.35}}
\put(4.5,3.7){\vector(2,-3){1.35}}
\put(6.0,3.7){\vector(2,-3){1.35}}
\put(11.7,7.0){$W_0$} 
\put(13.2,7.0){$W_1$} 
\put(14.7,7.0){$W_2$} 
\put(16.2,7.0){$W_3$} 
\put(17.7,7.0){$W_4$} 
\put(13.4,6.7){\vector(0,-1){2}}
\put(14.9,6.7){\vector(0,-1){2}}
\put(16.4,6.7){\vector(0,-1){2}}
\put(17.9,6.7){\vector(0,-1){2}}
\put(13.2,4.0){$Y_1$} 
\put(14.7,4.0){$Y_2$} 
\put(16.2,4.0){$Y_3$} 
\put(17.7,4.0){$Y_4$} 
\put(13.4,3.7){\vector(0,-1){2}}
\put(14.9,3.7){\vector(0,-1){2}}
\put(16.4,3.7){\vector(0,-1){2}}
\put(17.9,3.7){\vector(0,-1){2}}
\put(13.2,1.0){$X_1$} 
\put(14.7,1.0){$X_2$} 
\put(16.2,1.0){$X_3$} 
\put(17.7,1.0){$X_4$} 
\put(11.9,6.7){\vector(2,-3){1.35}}
\put(13.4,6.7){\vector(2,-3){1.35}}
\put(14.9,6.7){\vector(2,-3){1.35}}
\put(16.4,6.7){\vector(2,-3){1.35}}
\end{picture}
\end{center}
\vspace{-1cm}
\caption{Graphical representation of temporal dependence of order $p=1$. Type A (left), type B (right).}
\label{fig:graph} 
\end{figure}
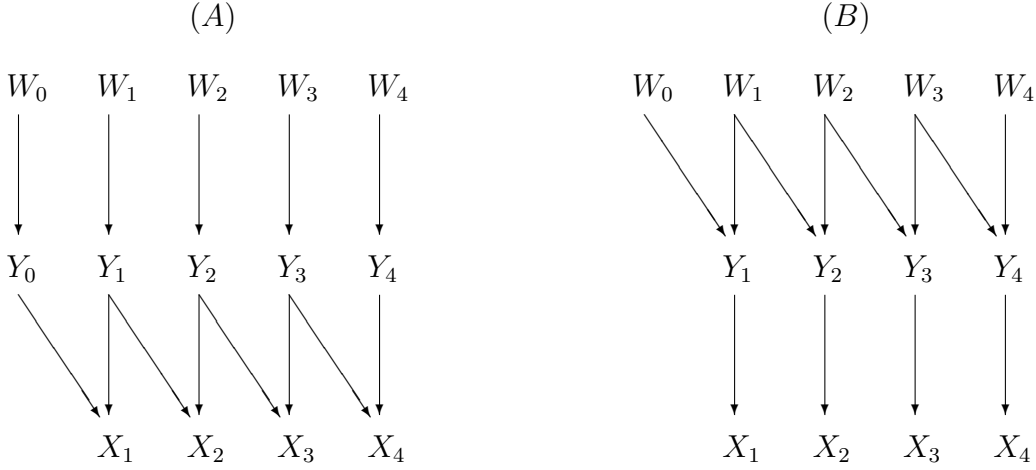

In general, the $W_t$'s will be independent Poisson random variables and the $Y_t$'s will be a binomial thinning of the $W_t$'s. 

\subsection{Type A dependence}
\label{ssec:typeA}

The stochastic process $\{X_t\}_{t\in\NN}$ is defined through the latent processes $\{Y_t\}_{t\in\TT}$ and $\{W_t\}_{t\in\TT}$, whose hierarchical representation is given by 
\begin{align}
\nonumber
W_t&\simiid\po(\mu), \\
\label{eq:typeA}
Y_t\mid W_t=w_t&\simind\bin(w_t,\alpha_t), \\
\nonumber
X_t-\sum_{i=0}^p y_{t-i}\mid \bY=\by&\simind\po\left(\mu\left(1-\sum_{i=0}^p\alpha_{t-i}\right)\right),
\end{align}
where $\mu>0$, $\alpha_t>0$ and $\sum_{i=0}^p\alpha_{t-i}<1$, for $t\in\NN$. 

Properties of the process $\{X_t\}_{t\in\NN}$, defined by type A construction \eqref{eq:typeA}, are given in Proposition \ref{prop:corA}. In particular, the marginal distribution and the autocorrelation function can be computed in closed form. 

\begin{proposition}
\label{prop:corA}
Let $\{X_t\}_{t\in\NN}$ be a stochastic process defined by equations \eqref{eq:typeA}. Then the marginal distribution of $X_t$ is $\po(\mu)$ for all $t\in\NN$, and the autocorrelation between $X_t$ and $X_{t+s}$ is given by
$$\Cor(X_t,X_{t+s})=\sum_{i=0}^{p-s}\alpha_{t-i},$$
for $1\leq s\leq p$ and zero for $s>p$. 
\end{proposition}
\begin{proof}
We note that the first level can be marginalised to keep only levels two and three. Since $W_t\sim\po(\mu)$ and $Y_t\mid W_t=w_t\sim\bin(w_t,\alpha_t)$, using \eqref{eq:main} we get $Y_t\sim\po(\mu\alpha_t)$ marginally and they are all independent across $t$. Then, since the sum of independent Poisson random variables is again Poisson, $\sum_{i=0}^p Y_{t-i}\sim\po\left(\mu\sum_{i=0}^p\alpha_{t-i}\right)$. Finally, considering level three of \eqref{eq:typeA} and using \eqref{eq:main}, we obtain $X_t\sim\po(\mu)$ marginally for $t\in\NN$. To obtain the correlation we rely on conditional independence properties and the iterative covariance formula. Then $\Cov(X_t,X_{t+s})=\E\{\Cov(X_t,X_{t+s}\mid\bY)\}+\Cov\{\E(X_t\mid\bY),\E(X_{t+s}\mid\bY)\}$. The first term in the sum becomes zero since $X_t$'s are conditional independent given $\bY$. The second term, after removing the constants of the expected values, is rewritten as $\Cov\left(\sum_{i=0}^pY_{t-i},\sum_{i=0}^pY_{t+s-i}\right)$. Since $Y_t$'s are independent, this covariance reduces to the variance of the common elements, that is, $\V\left(\sum_{i=0}^{p-s}Y_{t-i}\right)$. Again, since the random variable inside this variance is Poisson, we get that $\Cov(X_t,X_{t+s})=\mu\sum_{i=0}^{p-s}\alpha_{t-i}$. Finally, since $X_t$ and $X_{t+s}$ are $\po(\mu)$ marginally, the product of their standard deviations is $\mu$, so we obtain the result. 
\end{proof}

The autocorrelation expression of $\{X_t\}_{t\in\NN}$, given in Proposition \ref{prop:corA}, is a function of the sum of the thinning probabilities $\alpha_t$'s of the shared elements in the definition of $X_t$ and $X_{t+s}$. Additionally,  $\{X_t\}_{t\in\NN}$ becomes strictly stationary when $\alpha_t=\alpha$ for all $t$, and the autocorrelation induced reduces to $\Cor(X_t,X_{t+s})=(p-s+1)\alpha$. Moreover, if $p=0$, the $X_t$'s become independent. Alternatively, if $\alpha_t=0$ then $Y_t=0$ with probability one (w.p.1), so if $\alpha_t=0$ for all $t$, the $X_t$'s become also independent, regardless of the value of $p$. 

To see some similarities with previous proposals, we can re-write construction \eqref{eq:typeA} as 
\begin{equation}
\label{eq:typeA2}
X_t=\sum_{i=0}^pY_{t-i}+\epsilon_t=\sum_{i=0}^p \alpha_{t-i}\circ W_{t-i}+\epsilon_t,
\end{equation}
where $\epsilon_t\sim\po\left(\mu\left(1-\sum_{i=0}^p\alpha_{t-i}\right)\right)$. As such, \eqref{eq:typeA2} would resemble the Poisson MA$(q)$ given in \eqref{eq:maq} but with $p$ instead of $q$ and with an extra innovation term. However, the most important difference are the ``coefficients'' or thinning probabilities $\alpha_t$, which in our proposal they move along $t$, whereas in the MA$(q)$ they are fixed for any $t$.

\subsection{Type B dependence}
\label{ssec:typeB}

The stochastic process $\{X_t\}_{t\in\NN}$ is defined through the latent processes $\{Y_t\}_{t\in\TT}$ and $\{W_t\}_{t\in\TT}$, whose hierarchical representation is now given by 
\begin{align}
\nonumber
W_t&\simiid\po\left(\frac{\mu}{p+1}\right), \\
\label{eq:typeB}
Y_t\mid\bW=\bw&\simind\bin\left(\sum_{i=0}^p w_{t-i}\hspace{0.5mm},\alpha_t\right), \\
\nonumber
X_t-y_t\mid Y_t=y_t&\simind\po\left(\mu\left(1-\alpha_t\right)\right),
\end{align}
where $\mu>0$ and $\alpha_t\in(0,1)$, for $t\in\NN$. 

Properties of the process $\{X_t\}_{t\in\NN}$, defined by type B construction \eqref{eq:typeB}, are given in Proposition \ref{prop:corB}. As in type A construction, the marginal distribution and the autocorrelation function can be computed in closed form.
 
\begin{proposition}
\label{prop:corB}
Let $\{X_t\}_{t\in\NN}$ be a stochastic process defined by equations \eqref{eq:typeB}. Then the marginal distribution of $X_t$ is $\po(\mu)$ for all $t\in\NN$, and the autocorrelation between $X_t$ and $X_{t+s}$ is given by
$$\Cor(X_t,X_{t+s})=\alpha_t\alpha_{t+s}\left(\frac{p-s+1}{p+1}\right),$$
for $1\leq s\leq p$ and zero for $s>p$.
\end{proposition}
\begin{proof}
Using the additive property of independent Poisson variables, we obtain that \linebreak $\sum_{i=0}^pW_{t-i}\sim\po(\mu)$. Now, from \eqref{eq:main} and the second equation in \eqref{eq:typeB}, the marginal distribution of the latent variables $Y_t$'s becomes $Y_t\sim\po(\mu\alpha_t)$. Finally, from \eqref{eq:main} and the third equation in \eqref{eq:typeB}, we obtain that $X_t\sim\po(\mu)$ marginally for $t\in\NN$. Now for the correlation, we use the iterative covariance formula and apply conditional independence properties twice. We start with $\Cov(X_t,X_{t+s})=\E\{\Cov(X_t,X_{t+s}\mid\bY)\}+\Cov\{\E(X_t\mid\bY),\E(X_{t+s}\mid\bY)\}$. The first term in the sum is zero due to conditional independence of the $X_t$'s given $\bY$. The second term, after removing the constants in the expected values, becomes $\Cov(Y_t,Y_{t+s})$. Applying iterative covariance formula again we get $\Cov(Y_t,Y_{t+s})=\E\{\Cov(Y_t,Y_{t+s}\mid\bW)\}+\Cov\{\E(Y_t\mid\bW),\E(Y_{t+s}\mid\bW)\}$. Again, the first term becomes zero due to conditional independence of the $Y_t$'s given $\bW$, and computing the expected values in the second term we obtain $\Cov\{\alpha_t\sum_{i=0}^{p}W_{t-i},\alpha_{t+s}\sum_{i=0}^{p}W_{t+s-i}\}$. The $W_t$'s are independent, so after taking out the constants, this covariance reduces to the variance of the common elements, that is, $\alpha_t\alpha_{t+s}\V\left(\sum_{i=0}^{p-s}W_{t-i}\right)$. Since the random variable inside this variance is again Poisson, we get that $\Cov(X_t,X_{t+s})=\mu\alpha_t\alpha_{t+s}(p-s+1)/(p+1)$. Finally, since $X_t$ and $X_{t+s}$ are $\po(\mu)$ marginally, the product of their standard deviations is $\mu$, so we obtain the result. 
\end{proof}

The autocorrelation expression of $\{X_t\}_{t\in\NN}$, given in Proposition \ref{prop:corB}, is a function of the thinning probabilities of times $t$ and $t+s$, and the number of shared elements in the definition of $Y_t$ and $Y_{t+s}$. Again, $\{X_t\}_{t\in\NN}$ becomes strictly stationary when $\alpha_t=\alpha$ for all $t$, and the autocorrelation induced reduces to $\Cor(X_t,X_{t+s})=\alpha^2(p-s+1)/(p+1)$. Moreover, if $p=0$, the $X_t$'s become independent. Alternatively, if $\alpha_t=0$ then $Y_t=0$ w.p.1, so if $\alpha_t=0$ for all $t$, the $X_t$'s become also independent, regardless of the value of $p$. 

We note that the marginal distribution of the latent $Y_t$'s variables, in both type A and type B constructions, are the same, $Y_t\sim\po(\mu\alpha_t)$. However in \eqref{eq:typeA} they are independent, whereas in \eqref{eq:typeB} they are dependent.

Re-writing model \eqref{eq:typeB} into an additive form we have
\begin{equation}
\label{eq:typeB2}
X_t=Y_t+\epsilon_t=\alpha_t\circ\sum_{i=0}^p W_{t-i}+\epsilon_t,
\end{equation}
where $\epsilon_t\sim\po(\mu(1-\alpha_t))$. Expression \eqref{eq:typeB2} looks like a MA(0) process with innovation term, or like an INAR(1) process where the thinning operates over the sum of latent variables $\sum_{i=0}^p W_{t-i}$ instead of over the lagged variable $X_{t-1}$. 

Comparing the two constructions A and B in their re-written expressions \eqref{eq:typeA2} and \eqref{eq:typeB2}, disregarding the innovations $\epsilon_t$, type A process is based on the sum of $p+1$ thinnings of $p+1$ different latent variables, whereas type B process is based on a single thinning of the sum of $p+1$ latent variables. Therefore, model A should be more flexible for modelling purposes. 

To have an idea of how the paths of the processes look like, we simulated from both, type A and B processes, with $\mu=2$, $\alpha_t=1/7$, for $t=1,\ldots,T$ and $T=100$. We took three values of $p\in\{1,3,5\}$ to illustrate. Figure \ref{fig:sims} contains the simulated paths for both processes. For type A process (top row) there is a clear difference in the paths when we change the value of $p$, for $p=1$ (left panel) the process path shows a fast oscillation around the mean $\mu=2$, whereas as we increase $p$ (middle and right panels) the process paths start to oscillate more slowly around the mean. On the other hand, for type B process (bottom row), there is practically no difference in the paths when we increase the value of $p$. This is a results of the constant $\alpha_t$ parameters and the dependence imposed by type B construction.

\section{Bayesian inference}
\label{sec:inference}

Let $\bX=\{X_t,t=1,\ldots,T\}$ be an observable finite time series of integer-valued random variables. We assume that the law describing the sequence is one of the previously defined type A or type B models. The idea is to make inference about the unknown parameters of the models $\btheta=(\balpha,\mu)$, where $\balpha=\{\alpha_t,t=1,\ldots,T\}$, and for that we follow a Bayesian approach. 

For type A model, the parameter space is $\Theta_A=\{(\alpha_1,\ldots,\alpha_t,\mu):\alpha_t>0,\sum_{i=0}^p\alpha_{t-i}<1,t=1,\ldots,T,\mu>0\}$, and for type B model, the parameter space is $\Theta_B=\{(\alpha_1,\ldots,\alpha_t,\mu):\alpha_t\in(0,1),t=1,\ldots,T,\mu>0\}$. Given the flexibility of the parametric beta and gamma families to accommodate any prior knowledge, we use the former distribution for the parameters $\alpha_t$, $t=1,\ldots,T$, and the latter distribution for $\mu$. 

In summary, the prior distribution for the parameters $\btheta$ in both models is $$f(\btheta)=\left\{\prod_{t=1}^T\be(\alpha_t\mid a_\alpha,b_\alpha)\right\}\ga(\mu\mid a_\mu,b_\mu)I(\btheta\in\Theta_C),$$
where $C\in\{A,B\}$ for each of the two types of models, respectively. Note that the prior distribution $f(\btheta)$ for type A construction does not define independence for each of its components, because the parameter space $\Theta_A$ imposes a dependence in the $\balpha$ parameters, whereas for type B construction, prior $f(\btheta)$ imposes independence in all its components. 

To define the likelihood, we recall that $\bY$ and $\bW$ are latent variables, therefore are not observable, so we treat them as missing data and define an augmented likelihood (e.g. Tanner, 1991). For type A model the joint distribution of $(\bX,\bY)$, after integrating $\bW$ out, has the form
$$f(\bx,\by\mid\btheta)=\prod_{t=1}^T\po\left(x_t-\sum_{i=0}^p y_{t-i}\left|\mu\left(1-\sum_{i=0}^p\alpha_{t-i}\right)\right.\right)\po(y_t\mid\mu\alpha_t);$$
and for type B model the joint distribution of $(\bX,\bY,\bW)$ has the form
$$f(\bx,\by,\bw\mid\btheta)=\prod_{t=1}^T \po\left(x_t-y_t\mid\mu(1-\alpha_t)\right)\bin\left(y_t\left|\sum_{i=0}^p w_{t-i},\alpha_t\right.\right)\po\left(w_t\left|\frac{\mu}{p+1}\right.\right).$$

Posterior distributions of $\btheta$ are simply proportional to the product of the augmented likelihoods by the prior. These will be characterised through their full conditional distributions, which have been included in the Appendix for both types of models. Distributions (i)--(iii) correspond to type A model, whereas distributions (iv)--(vii) correspond to type B model. Posterior inference is therefore obtained through the implementation of a Gibbs sampler \citep{smith&roberts:93} with some Metropolis-Hastings (MH) steps \citep{tierney:94}. Details are also given in the Appendix.

\section{Numerical analysis}
\label{sec:numerical}

Unfortunately maternal mortality is still an important public health problem in Mexico. According to the World Health Organization, maternal mortality is defined as a death from preventable causes related to pregnancy and childbirth. The Mexican National Institute of Geography and Statistics reports the annual number of maternal deaths for the 32 political states of Mexico (https://www.inegi.org.mx/sistemas/olap/proyectos/bd/continuas/mor\-ta\-li\-dad/mortalidadgeneral.asp). Information is available from 1990 until 2018, that is, a total of $T=29$ years. This dataset is provided as a supplementary material. 

Along available years, the states with the smallest number of deaths are, Baja California Sur and Colima, with an average of 3.4 and 3.5 deaths per year, respectively. On the opposite extreme, the states with the largest number of deaths are CDMX (Mexico City) and the State of Mexico, with an average of 147 and 139 deaths across the states. It is not surprising that the states with the smallest and largest number of deaths correspond to the least and the most populated states, respectively. On the other hand, across states,  2018 is the year with the smallest number of deaths, with and average of 28, and 1990 is the year with the largest number of deaths, with ad average of 46. This suggests an overall reduction in the number of deaths along years. 

We analysed the 32 time series with both types of models. To define the prior distributions we took $a_\alpha=b_\alpha=a_\mu=b_\mu=0.01$, which define vague priors (large variance) for $\alpha_t$, $t=1,\ldots,T$ and $\mu$. For $p$ we took a set of different values to compare, say $p\in\{0,1,2,3,4,5,6\}$. A Gibbs sampler was implemented in Fortran with 16,000 iterations, a burn-in period of 1,000 and kept one of every 5$th$ iteration, after burn-in, to produce posterior summaries. For each state the running time is less than 10 seconds. The tuning parameters for the MH steps were set to $\delta_\alpha=3$ and $\delta_w=10$ that provide acceptance probabilities between 20\% and 40\%, which according to \cite{robert&casella:10} are optimal. Convergence of the chains was assessed informally by looking at the trace plots, ergodic means and autocorrelation functions. Figure \ref{fig:mcmc} shows these convergence diagnostics for parameter $\mu$ in type A model for Coahuila state. 

To assess model fit we computed the L-measure which is a predictive statistic that summarises variance and mean square error (bias) of the posterior predictive distribution of each $X_t$. This is defined as \citep{ibrahim&laud:94}
\begin{equation}
\label{eq:lmeasure}
L(\nu)=\frac{1}{T}\sum_{t=1}^T\V\left(X_t^F\mid\bx\right)+\frac{\nu}{T}\sum_{t=1}^T\left\{\E\left(X_t^F\mid\bx\right)-x_t\right\}^2,
\end{equation}
where $X_t^F$ and $x_t$ denote the predictive and observed value of $X_t$, respectively.

Table \ref{tab:gof} reports the values of the L-measure with $\nu=1/2$, obtained when fitting models of types A and B to the 32 time series of the maternal mortality dataset, for $p=0,1,\ldots,6$. For each type of model, the value of $p$ with the smallest L-measure is highlighted in bold. Apart from Aguascalientes and Zacatecas (see Table \ref{tab:gof}), where the best fitting is achieved for $p=0$ (independence) in one of the two types of models, for the rest of the states the best fitting model is obtained for $p>0$, which implies a temporal dependence. Now, comparing the best fitting from the two types, for 31 of the 32 states, type A model outperforms type B model. The only state where type B model is slightly better is Colima with an L-measure of 2.81 as compared to 2.84 obtained by the best type A model. 

Figures \ref{fig:baja}, \ref{fig:coahuila} and \ref{fig:cdmx} show the performance of best fitting models for type A (left panel) and type B (right panel) for Baja California, Coahuila and CDMX (Mexico City), respectively. In these figures, type A model shows a better fitting than type B model, with more accurate predictions and narrower 95\% credible intervals. On the other hand, Figure \ref{fig:colima} displays model performance for the state of Colima, which is the only case where type B model slightly outperforms type A model. 

Finally, to place our two proposals in context, we fitted two commonly used models for integer valued time series: the INAR(1) model \eqref{eq:inar1}, with prior distributions $\alpha\sim\be(0.01,0.01)$ and $\mu\sim\ga(0.01,0.01)$ independently; and the INGARCH(1,1) model defined as $X_t\sim\po(\mu_t)$ and $\log(\mu_t)=\alpha+\beta_1\log(\mu_{t-1})+\beta_2\log(X_{t-1}+1)$, with prior distributions $\alpha\sim\no(0,0.01)$ and $\beta_j\sim\no(0,0.01)$, for $j=1,2$ independently. We also implemented Gibbs samplers with the same specifications as above and computed the L-measure \eqref{eq:lmeasure} with $\nu=1/2$. The corresponding goodness of fit statistics are included in the last two columns of Table \ref{tab:gof} for the 32 states of Mexico. Interestingly, for 26 of the 32 states our best fitting type A model outperforms the INAR(1) model, and for all states our best fitting type A model is better than the INGARCH(1,1) model. 

To compare the performance of our two processes with the two chosen competitors, we show in Figure \ref{fig:gua} the fittings for the state of Guanajuato. In each of these graphs we further include out of sample predictions for 3 years ahead. Future predictions with type A model (top left) are the only ones that follow the decreasing tendency of the data, whereas for type B model (top right) and INAR(1) model (bottom left) future predictions are slightly increasing, finally for INGARCH(1,1) model out of sample predictions remain fairly constant.

\section{Extensions}
\label{sec:extensions}

Considering Figure \ref{fig:graph}, we note that the processes $\{X_t\}_{t\in\NN}$ are still well defined if any of the diagonal arrows are removed. So in general, we can make $X_t$ to be defined in terms of $Y_{t-i}$, in type A construction, or $Y_t$ to be defined in terms of $W_{t-i}$, in type B construction, for any $i$ not necessarily consecutive. Therefore we can define more general seasonal \citep{nabeya:01} or periodic \citep{mcleod:94} dependent models.

Let $\{X_t\}_{t\in\NN}$ be a stochastic process with seasonality $s$, and let $\bY=\{Y_t\}_{t\in\TT_s}$ and $\bW=\{W_t\}_{t\in\TT_s}$ be two latent processes with $\TT_s=\{1-ps,,-ps,-ps+1,\ldots\}$. Then a seasonal dependent process of order $p$ would be defined by
$$X_t-\sum_{i=0}^p y_{t-si}\mid\bY=\by\simind\po\left(\mu\left(1-\sum_{i=0}^p \alpha_{t-si}\right)\right),$$
for a type A construction, with levels 1 and 2 as in \eqref{eq:typeA}, and with parameter constraint $\sum_{i=0}^p \alpha_{t-si}<1$ for $t\in\NN$; and
$$Y_t\mid\bW=\bw\simind\bin\left(\sum_{i=0}^p w_{t-si},\alpha_t\right),$$
for a type B construction, with levels 1 and 3 as in \eqref{eq:typeB}. In both types, an analogous proof to Propositions \ref{prop:corA} and \ref{prop:corB}, would show that $X_t\sim\po(\mu)$ marginally for $t\in\NN$. 

Now, if $\{X_t\}_{t\in\NN}$ is a process such that we can re-write the time index as $t=t(r,m)=(r-1)s+m$, for $r=1,2,\ldots$ and $m=1,\ldots,s$, we can define a periodic dependent process of orders $(p_1,\ldots,p_s)$. For instance, for monthly data, $s=12$ and $r$ and $m$ denote the year and month, respectively. Let $\bY=\{Y_t\}_{t\in\TT_s}$ and $\{W_t\}_{t\in\TT_s}$ be two latent processes with $\TT_s=\{t^*,t^*+1,\ldots\}$ and $t^*=\min\{t(r,m)-p_m:r=1,m=1,\ldots,s\}$. Then a periodic dependent process of orders $(p_1,\ldots,p_s)$ would be defined by
$$X_t-\sum_{i=0}^{p_m}y_{t(r,m)-i}\mid\bY=\by\simind\po\left(\mu\left(1-\sum_{i=0}^{p_m}\alpha_{t(r,m)-i}\right)\right)$$
for a type A construction with levels 1 and 2 as in \eqref{eq:typeA} and with $\sum_{i=0}^{p_m}\alpha_{t(r,m)-i}<1$ for $t=t(r,m)\in\NN$; and
$$Y_t\mid\bW=\bw\simind\bin\left(\sum_{i=0}^{p_m}w_{t(r,m)-i}\hspace{0.5mm},\alpha_t\right)$$
for a type B construction with levels 1 and 3 as in \eqref{eq:typeB}. It is not difficult to prove that for type A construction we obtain $X_t\sim\po(\mu)$ marginally for all $t\in\NN$, whereas for type B construction we obtain $X_t\sim\po\left(\mu\left\{1-\alpha_t+\alpha_t(p_m+1)/(p+1)\right\}\right)$ marginally for $t\in\NN$. 

Alternatively, both constructions can also be suitably defined for a spatial setting.
Let $\{X_t\}_{t\in\SC}$ be a stochastic process and assume that the index $t$ denotes spatial location in the set $\SC=\{1,\ldots,n\}$, and consider $\partial_t$ to be the set of neighbours of location $t$. Let $\bY=\{Y_t\}_{t\in\SC}$ and $\bW=\{W_t\}_{t\in\SC}$ be two latent processes. Then, a spatial dependent process $\{X_t\}_{t\in\SC}$ would be defined by
$$X_t-\sum_{i\in\partial_t} y_{i}\mid\bY=\by\simind\po\left(\mu\left(1-\sum_{i\in\partial_t} \alpha_{i}\right)\right),$$
for a type A construction with levels 1 and 2 as in \eqref{eq:typeA} and with $\sum_{i\in\partial_t} \alpha_{i}<1$; and
$$Y_t\mid\bW=\bw\simind\bin\left(\sum_{i\in\partial_t} w_{i},\alpha_t\right),$$
for a type B construction with levels 1 and 3 as in \eqref{eq:typeB}. Again, only in type A construction we obtain $X_t\sim\po(\mu)$ marginally for all $t\in\SC$.

Furthermore, combinations of any temporal with spatial dependences are also possible by an appropriate definition of the sums.

\section{Concluding remarks}
\label{sec:conclude}

We have introduced two novel ways of defining dependence, in space and time, among Poisson random variables. Our proposal relies on the use of latent variables in a three levels hierarchical model. Both constructions have shown a good performance when modelling real datasets, with an advantage for type A model over type B model, for the specific maternal mortality dataset analysed here. Additionally, our models outperformed the most commonly used INAR(1) and INGARCH(1,1) models in the maternal mortality dataset. 

When using our proposals for modelling purposes, one has to be aware of their different features. Type B construction induces a correlation, given in Proposition \ref{prop:corB}, that only depends on two parameters. On the other hand, type A construction induces a more flexible autocorrelation, see Proposition \ref{prop:corA}, in the sense that it could be based on several more parameters. 

A straightforward generalisation of our proposal is to define stochastic processes with negative binomial marginal distributions. Considering that \citep[e.g.][p.141]{nieto&bandyopadhyay:13}, if $X\mid Z=z\sim\po(z)$ and $Z\sim\ga(a,b)$, then $X\sim\pg(a,b,1)$, that is, a Poisson-gamma distribution with mean $a/b$. Furthermore, if $a$ is an integer $\pg(a,b,1)\equiv\nbi(a,b/(b+1))$, that is, a negative binomial distribution with number of successes $a$ and probability of success $b/(b+1)$. Therefore if for type A construction \eqref{eq:typeA} or for type B construction \eqref{eq:typeB}, we assume that the level 3 equations are given conditionally on $\mu$, and if we further take $\mu\sim\ga(r,\pi/(1-\pi))$ then $X_t\sim\nbi(r,\pi)$ marginally for $t\in\NN$. Studying the performance of the negative binomial processes is left to study in a future work. 

Finally, our constructions are flexible enough to be used in different contexts.  For the maternal mortality dataset, our models were used as sampling models to describe the law of the data. However, they can also be used as prior distributions for discrete functional integer-valued parameters, in a Bayesian nonparametric analysis.

\section*{Appendix}
Full conditional distributions for model parameters $\btheta$ and latent variables $(\bY,\bW)$ to perform posterior inference for type A and type B models. For simplicity we assume that $Y_t=0$, $W_t=0$ and $\alpha_t=0$ for $t\leq0$. 
In the sequel, we use $I_{\XX}(x)$ to denote the indicator function that takes the value of one if $x\in\XX$ and zero otherwise.

For type A model, the required full conditional distributions are:
\begin{enumerate}
\item[i)] For $Y_t$, $t=1,\ldots,T$
$$f(y_t\mid\rest)\propto\frac{\left[\alpha_t\mu^{-p}\left\{\prod_{j=0}^p\left(1-\sum_{i=0}^p\alpha_{t+j-i}\right)\right\}^{-1}\right]^{y_t}}{y_t!\prod_{j=0}^p\left(x_{t+j}-\sum_{i=0}^p y_{t+j-i}\right)!}I_{\{0,\ldots,c_t\}}(y_t),$$
with $c_t=\min_{j=0,\ldots,p}\{x_{t+j}-\sum_{i=0,i\neq j}^p y_{t+j-i}\}$
\item[ii)] For $\alpha_t$, $t=1,\ldots,T$
$$f(\alpha_t\mid\rest)\propto\alpha_t^{a_\alpha+y_t-1}(1-\alpha_t)^{b_\alpha-1}e^{p\mu\alpha_t}\prod_{j=0}^p\left(1-\sum_{i=0}^p\alpha_{t+j-i}\right)^{x_{t+j}-\sum_{i=0}^p y_{t+j-i}}I_{(0,d_t)}(\alpha_t)$$
where $d_t=\min_{j=0,\ldots,p}\left\{1-\sum_{i=0,i\neq j}^p\alpha_{t+j-i}\right\}$
\item[iii)] For $\mu$
$$f(\mu\mid\rest)=\ga\left(\mu\left|a_\mu+\sum_{t=1}^T x_t-\sum_{t=1}^T\sum_{i=1}^p y_{t-i},b_\mu+T+\sum_{t=1}^T\sum_{i=1}^p\alpha_{t-i}\right.\right)$$
\end{enumerate}

Since (i) is a discrete distribution with bounded support, we simply evaluate at all points of the support and normalize to obtain the probability density and sample a new $y_t^{(l)}$ at iteration $l$. To sample from (ii) we implement a MH step with random walk proposal distribution. If $\alpha_t^{(l)}$ is the current state of the chain, we sample from $\alpha_t^*\mid\alpha_t^{(l)}\sim\un(\max(0,\alpha_t^{(l)}-\delta_\alpha,\min(d_t,\alpha_t^{(l)}+\delta_\alpha)))$, that is a continuous uniform distribution, and accept it with probability $\min\{1,f(\alpha_t^*\mid\rest)/f(\alpha_t^{(l)}\mid\rest)\}$. Sampling from (iii) is direct since it has a standard form. 

For type B model, the required full conditional distributions are:
\begin{enumerate}
\item[iv)] For $Y_t$, $t=1,\ldots,T$
$$f(y_t\mid\rest)\propto\frac{\left\{\alpha_t\mu^{-1}(1-\alpha_t)^{-2}\right\}^{y_t}}{(x_t-y_t)!y_t!\left(\sum_{i=0}^p w_{t-i}-y_t\right)!}I_{\{0,\ldots,m_t\}}(y_t),$$
with $m_t=\min\{x_{t},\sum_{i=0}^p w_{t-i}\}$
\item[v)] For $W_t$, $t=1,\ldots,T$
$$f(w_t\mid\rest)\propto \left\{\prod_{j=0}^p {{\sum_{i=0}^p w_{t+j-i}}\choose{y_{t+j}}}\right\}\left\{\frac{\mu}{p+1}\prod_{j=0}^p(1-\alpha_{t+j})\right\}^{w_t}\frac{1}{w_t!}I_{\{h_t,h_t+1\ldots,\}}(w_t),$$
where $h_t=\max_{j=0,\ldots,p}\{y_{t+j}-\sum_{i=0,i\neq j}^p w_{t+j-i}\}$
\item[vi)] For $\alpha_t$, $t=1,\ldots,T$
$$f(\alpha_t\mid\rest)\propto\alpha_t^{a_\alpha+y_t-1}(1-\alpha_t)^{b_\alpha+x_t+\sum_{i=0}^p w_{t-i}-2y_t-1}e^{\mu\alpha_t}I_{(0,1)}(\alpha_t)$$
\item[vii)] For $\mu$
$$f(\mu\mid\rest)=\ga\left(\mu\left|a_\mu+\sum_{t=1}^T (x_t+w_t-y_t),b_\mu+T\left(\frac{p+2}{p+1}\right)-\sum_{t=1}^T\alpha_{t}\right.\right)$$
\end{enumerate}

Again, since (iv) is a discrete distribution with bounded support, we proceed as for (i). To sample from (v) we note that the support is discrete but unbounded, so we implement a MH step with random walk proposal of the form $W_t^*\mid W_t^{(l)}=w_t^{(l)}\sim\un(\max(h_t,w_t^{(l)}-\delta_w),w_t^{(l)}+\delta_w)$ and accept it with probability $\min\{1,f(w_t^*\mid\rest)/f(w_t^{(l)}\mid\rest)\}$. To sample from (vi) we proceed as for (ii) but with proposal $\alpha_t^*\mid\alpha_t^{(l)}\sim\un(\max(0,\alpha_t^{(l)}-\delta_\alpha,\min(1,\alpha_t^{(l)}+\delta_\alpha)))$. Finally, sampling from (vii) is direct. In all cases, $\delta_\alpha$ and $\delta_w$ are tuning parameters that control the acceptance probability.

\section*{Acknowledgements}

The author acknowledges support from \textit{Asociaci\'on Mexicana de Cultura, A.C.}

\bibliographystyle{natbib}

\newpage

\begin{table}
\centering
{\tiny
\begin{tabular}{l|ccccccc|ccccccc|cc}
\toprule
State / & \multicolumn{7}{c|}{Type A} & \multicolumn{7}{c|}{Type B} & AR & GA- \\
\hspace{1.1cm}$p$	&	0	&	1	&	2	&	3	&	4	&	5	&	6	&	0	&	1	&	2	&	3	&	4	&	5	&	6	&	&	RCH\\
\midrule
{\tiny Aguascal.}	&	{\bf 9.3}	&	11.3	&	13.1	&	13.5	&	11.3	&	10.8	&	11.1	&	{\bf 12.2}	&	12.5	&	12.9	&	13.7	&	13.8	&	12.7	&	13.9	&	4.1	&	19.8\\
{\tiny Baja.Calif.}	&	29.2	&	15.4	&	11.2	&	14.3	&	{\bf 10.0}	&	14.2	&	12.3	&	41.5	&	21.8	&	20.5	&	21.4	&	22.5	&	{\bf 17.9}	&	19.8	&	14.2	&	38.3\\
{\tiny Baja.Cal.S.}	&	3.3	&	3.0	&	2.7	&	2.6	&	3.0	&	2.9	&	{\bf 2.2}	&	3.3	&	3.3	&	{\bf 3.2}	&	3.7	&	3.9	&	3.3	&	3.5	&	1.3	&	5.5\\
{\tiny Campeche}	&	5.1	&	6.4	&	5.8	&	4.0	&	{\bf 2.6}	&	3.1	&	5.6	&	6.2	&	6.5	&	6.9	&	7.0	&	7.2	&	6.1	&	{\bf 5.8}	&	3.8	&	10.4\\
{\tiny Coahuila}	&	21.6	&	27.9	&	12.0	&	13.2	&	{\bf 8.8}	&	11.7	&	9.4	&	28.1	&	18.5	&	17.2	&	17.8	&	17.0	&	{\bf 15.2}	&	15.8	&	9.7	&	27.1\\
{\tiny Colima}	&	3.2	&	{\bf 2.8}	&	3.5	&	2.9	&	3.7	&	3.6	&	3.5	&	3.4	&	3.4	&	3.7	&	3.9	&	{\bf 2.8}	&	3.5	&	3.8	&	2.1	&	5.6\\
{\tiny Chiapas}	&	19.0	&	22.7	&	11.6	&	14.2	&	{\bf 9.3}	&	15.6	&	15.6	&	55.5	&	29.6	&	24.6	&	25.7	&	{\bf 20.3}	&	22.4	&	22.4	&	28.9	&	38.1\\
{\tiny Chihuahua}	&	19.1	&	20.2	&	14.5	&	11.9	&	8.1	&	10.1	&	{\bf 6.8}	&	26.0	&	{\bf 15.3}	&	18.3	&	15.3	&	17.3	&	17.1	&	17.3	&	14.1	&	31.7\\
{\tiny CDMX}	&	25.1	&	17.5	&	11.7	&	11.5	&	{\bf 6.6}	&	12.9	&	13.2	&	78.8	&	43.8	&	33.0	&	{\bf 20.6}	&	24.3	&	29.5	&	25.2	&	38.5	&	49.5\\
{\tiny Durango}	&	15.6	&	20.3	&	16.7	&	16.4	&	15.5	&	{\bf 13.3}	&	15.0	&	19.5	&	17.4	&	18.3	&	{\bf 16.0}	&	16.8	&	18.0	&	16.2	&	10.4	&	27.6\\
{\tiny Guanajuato}	&	32.1	&	42.8	&	{\bf 10.4}	&	30.0	&	16.9	&	26.0	&	21.9	&	61.8	&	34.8	&	{\bf 27.7}	&	32.4	&	29.5	&	29.2	&	34.1	&	46.0	&	43.5\\
{\tiny Guerrero}	&	19.7	&	18.1	&	10.0	&	7.4	&	4.3	&	9.6	&	{\bf 5.8}	&	37.7	&	15.4	&	18.2	&	17.0	&	{\bf 14.9}	&	16.4	&	17.5	&	16.7	&	36.4\\
{\tiny Hidalgo}	&	32.8	&	37.4	&	{\bf 10.0}	&	24.1	&	20.8	&	23.9	&	21.5	&	71.4	&	42.2	&	{\bf 32.5}	&	36.4	&	34.9	&	37.2	&	35.8	&	40.0	&	51.0\\
{\tiny Jalisco}	&	15.5	&	27.6	&	22.8	&	12.8	&	{\bf 9.8}	&	24.0	&	19.4	&	58.1	&	27.4	&	24.8	&	27.2	&	26.1	&	26.7	&	{\bf 23.6}	&	24.9	&	41.7\\
{\tiny Mexico}	&	22.7	&	10.6	&	9.7	&	{\bf 5.7}	&	7.2	&	7.1	&	8.0	&	56.9	&	27.0	&	21.9	&	22.0	&	19.7	&	{\bf 16.7}	&	20.6	&	30.0	&	38.1\\
{\tiny Michoacan}	&	16.9	&	20.5	&	{\bf 8.8}	&	9.8	&	10.1	&	9.2	&	13.9	&	36.2	&	21.0	&	19.7	&	20.6	&	{\bf 18.9}	&	21.4	&	21.3	&	16.3	&	34.5\\
{\tiny Morelos}	&	19.9	&	22.3	&	13.3	&	14.8	&	{\bf 12.4}	&	16.3	&	13.2	&	27.3	&	21.9	&	20.8	&	{\bf 16.4}	&	20.1	&	17.0	&	18.9	&	13.0	&	36.5\\
{\tiny Nayarit}	&	10.7	&	12.1	&	9.9	&	{\bf 6.7}	&	8.5	&	7.9	&	8.9	&	10.3	&	10.9	&	11.6	&	10.5	&	{\bf 9.5}	&	10.9	&	10.5	&	5.3	&	16.8\\
{\tiny Nuevo.Leon}	&	32.0	&	20.0	&	18.7	&	{\bf 11.4}	&	16.9	&	37.0	&	26.6	&	66.8	&	40.9	&	38.1	&	{\bf 22.8}	&	27.6	&	30.4	&	33.0	&	17.3	&	49.0\\
{\tiny Oaxaca}	&	24.6	&	15.4	&	16.7	&	{\bf 12.8}	&	23.0	&	19.9	&	14.9	&	66.5	&	44.3	&	28.9	&	30.1	&	{\bf 28.7}	&	30.8	&	32.6	&	41.2	&	49.2\\
{\tiny Puebla}	&	54.2	&	26.1	&	{\bf 9.4}	&	18.7	&	15.9	&	9.5	&	15.2	&	83.7	&	45.0	&	28.6	&	{\bf 28.3}	&	30.3	&	30.6	&	38.7	&	45.2	&	51.6\\
{\tiny Queretaro}	&	30.1	&	17.3	&	17.9	&	{\bf 10.7}	&	15.8	&	14.3	&	32.8	&	39.9	&	23.5	&	26.2	&	{\bf 21.2}	&	25.0	&	28.1	&	29.1	&	18.5	&	45.0\\
{\tiny Quintana.R.}	&	9.2	&	9.2	&	5.9	&	{\bf 4.6}	&	7.6	&	5.5	&	6.2	&	8.4	&	{\bf 8.3}	&	8.8	&	8.6	&	10.3	&	10.2	&	10.8	&	4.8	& 15.7\\
{\tiny San.Luis.P.}	&	50.4	&	29.1	&	14.2	&	{\bf 13.4}	&	16.5	&	18.0	&	25.0	&	97.8	&	60.1	&	47.1	&	42.8	&	47.2	&	47.0	&	{\bf 41.4}	&	40.6	&	71.1\\
{\tiny Sinaloa}	&	19.6	&	21.5	&	12.6	&	17.0	&	{\bf 11.5}	&	11.6	&	12.8	&	23.2	&	19.0	&	18.3	&	16.8	&	18.7	&	{\bf 16.2}	&	21.5	&	9.4	&	31.7\\
{\tiny Sonora}	&	13.6	&	18.0	&	6.9	&	7.1	&	13.5	&	{\bf 6.4}	&	6.4	&	22.7	&	16.7	&	17.0	&	19.4	&	17.6	&	{\bf 15.4}	&	17.3	&	10.5	&	27.1\\
{\tiny Tabasco}	&	27.2	&	23.8	&	18.1	&	{\bf 9.3}	&	21.5	&	25.8	&	18.3	&	54.7	&	31.0	&	32.0	&	31.9	&	{\bf 26.2}	&	27.7	&	27.5	&	13.9	&	43.3\\
{\tiny Tamaulipas}	&	42.2	&	27.9	&	{\bf 14.5}	&	25.0	&	25.2	&	22.6	&	21.4	&	76.5	&	48.5	&	41.9	&	35.8	&	{\bf 24.3}	&	34.6	&	37.9	&	15.8	&	55.9\\
{\tiny Tlaxcala}	&	11.9	&	14.4	&	7.2	&	7.5	&	10.2	&	{\bf 6.2}	&	9.9	&	14.2	&	15.2	&	{\bf 10.1}	&	12.4	&	12.7	&	11.9	&	13.9	&	11.4	&	20.9\\
{\tiny Veracruz}	&	16.5	&	21.9	&	14.6	&	12.8	&	8.9	&	9.3	&	{\bf 7.6}	&	45.0	&	24.7	&	{\bf 18.1}	&	23.0	&	21.4	&	21.0	&	19.8	&	33.3	&	37.6\\
{\tiny Yucatan}	&	25.4	&	29.0	&	12.2	&	{\bf 9.2}	&	17.1	&	9.7	&	15.7	&	41.9	&	{\bf 16.7}	&	19.3	&	20.3	&	22.4	&	25.8	&	26.7	&	17.2	&	39.0\\
{\tiny Zacatecas}	&	12.1	&	15.3	&	8.4	&	{\bf 5.8}	&	6.6	&	7.0	&	7.7	&	{\bf 8.9}	&	11.7	&	12.9	&	13.9	&	13.9	&	14.8	&	15.4	&	9.6	&	21.4\\
\bottomrule
\end{tabular}}
\caption{L-measure \eqref{eq:lmeasure} with $\nu=1/2$ when fitting type A and B models for $p=0,...,6$ to maternal mortality data for the 32 states of Mexico. Smallest value, within each model type is shown in bold. L-measure for INAR(1) and INGARCH(1,1) models are also included in the last two columns.}
\label{tab:gof}
\end{table}

\begin{figure}
\centerline{\includegraphics[scale=0.27]{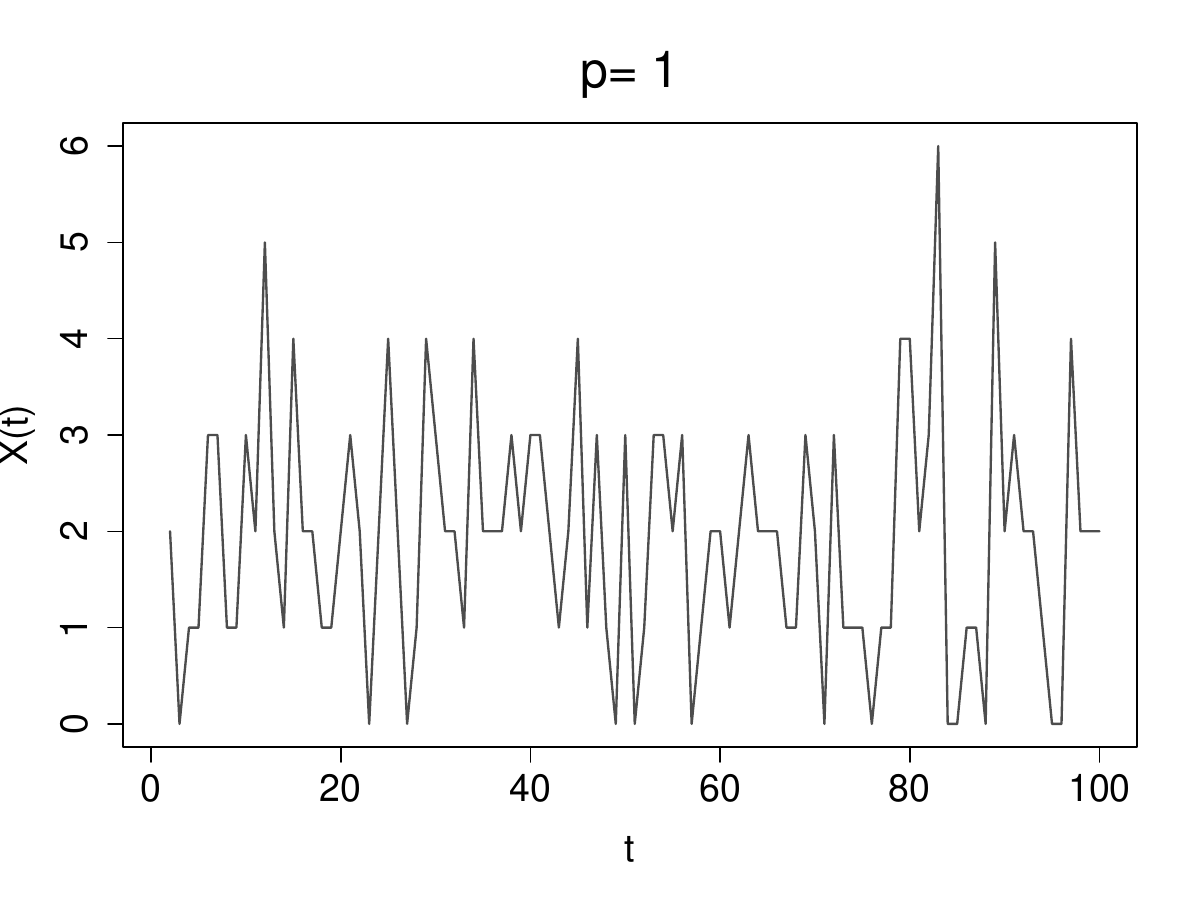}
\includegraphics[scale=0.27]{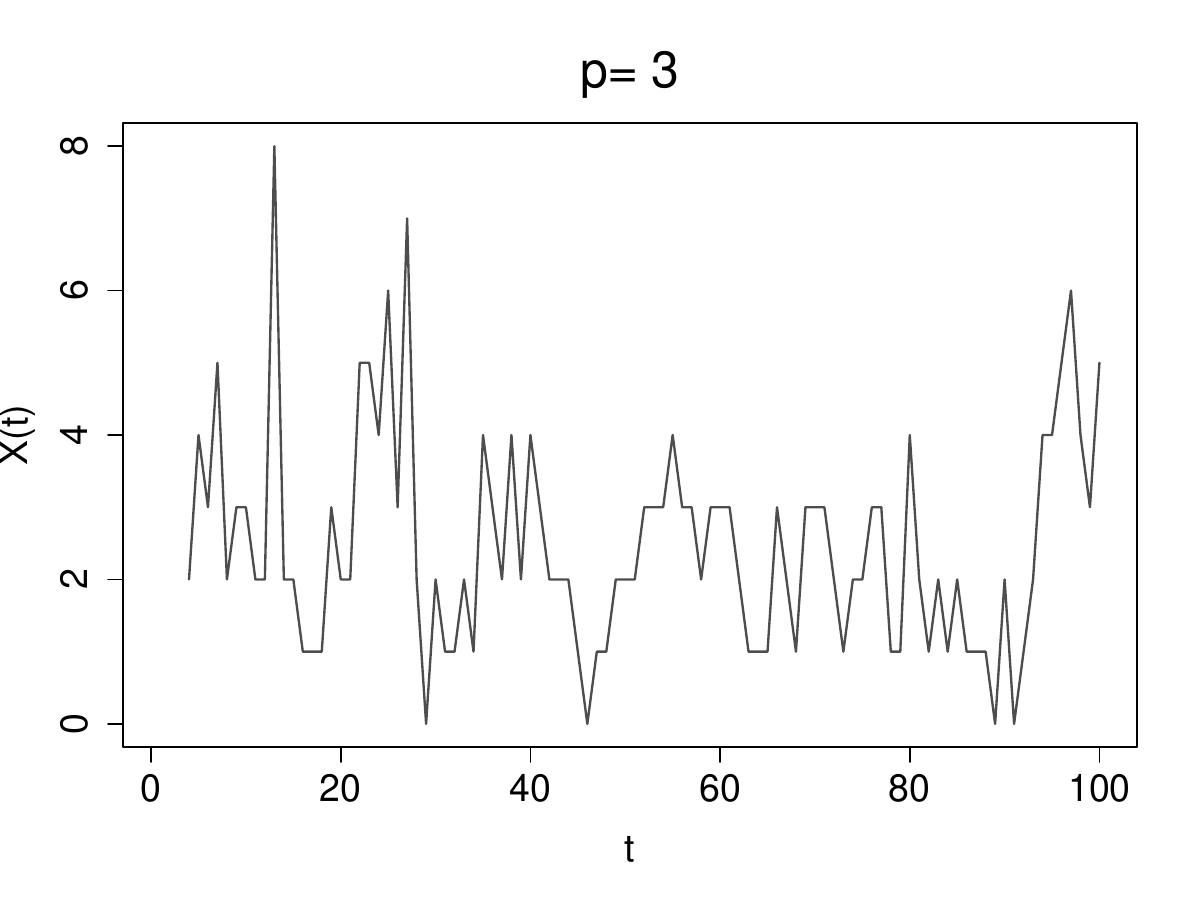}
\includegraphics[scale=0.27]{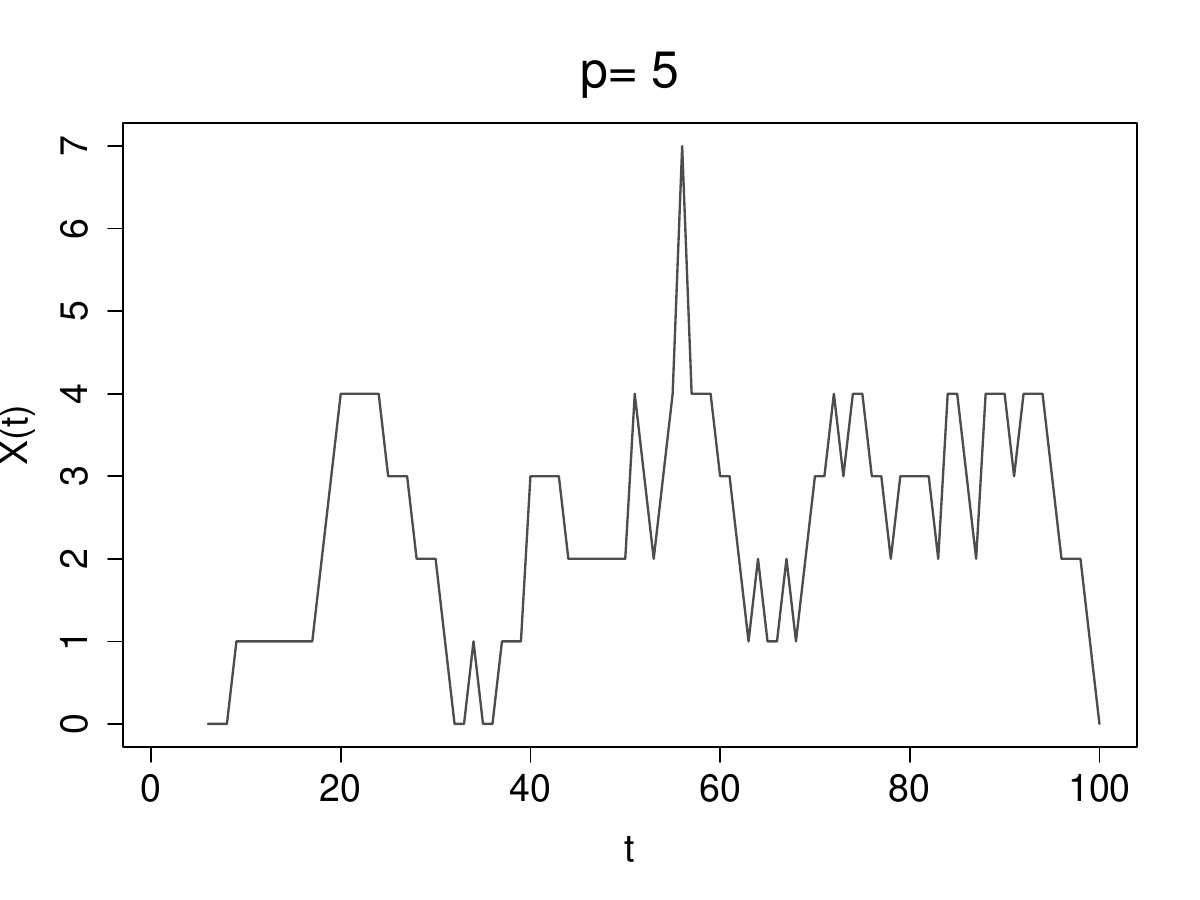}}
\centerline{\includegraphics[scale=0.27]{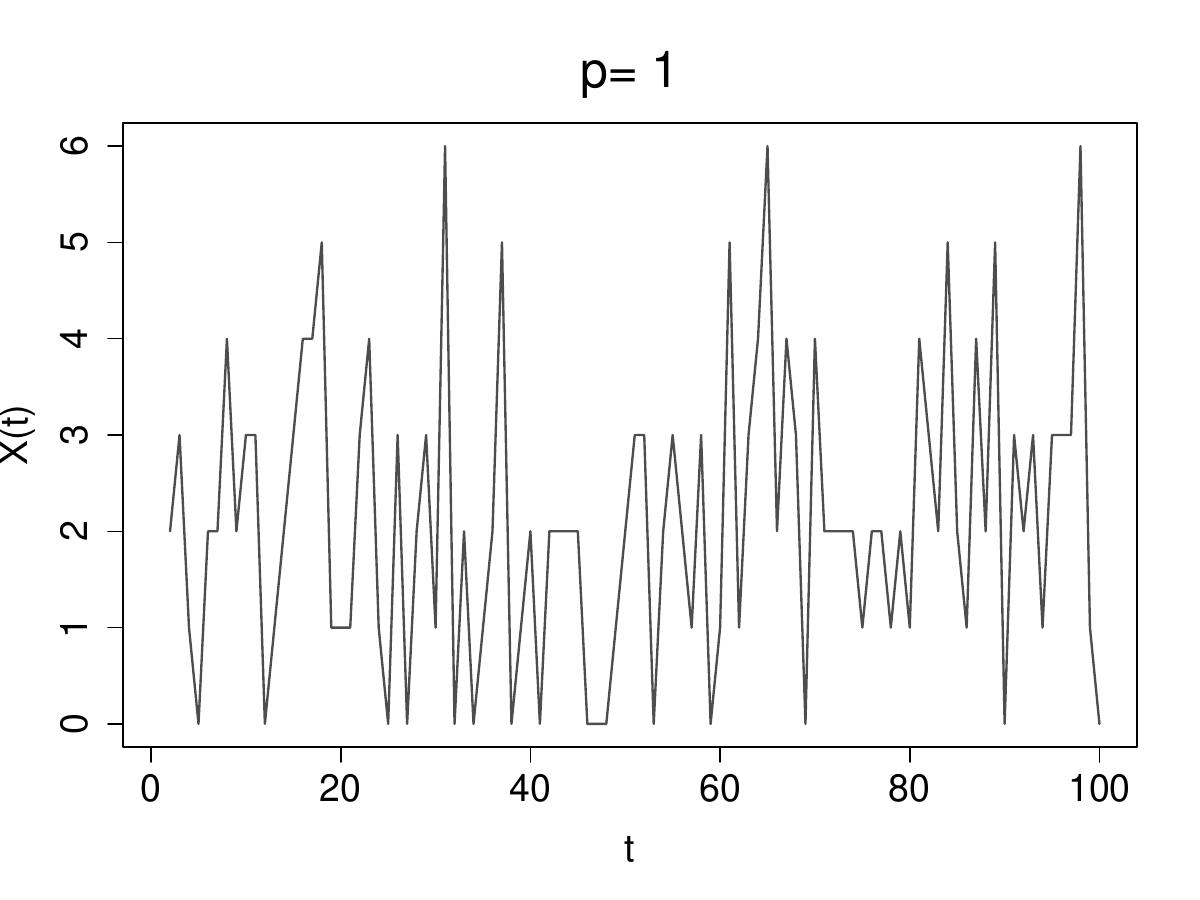}
\includegraphics[scale=0.27]{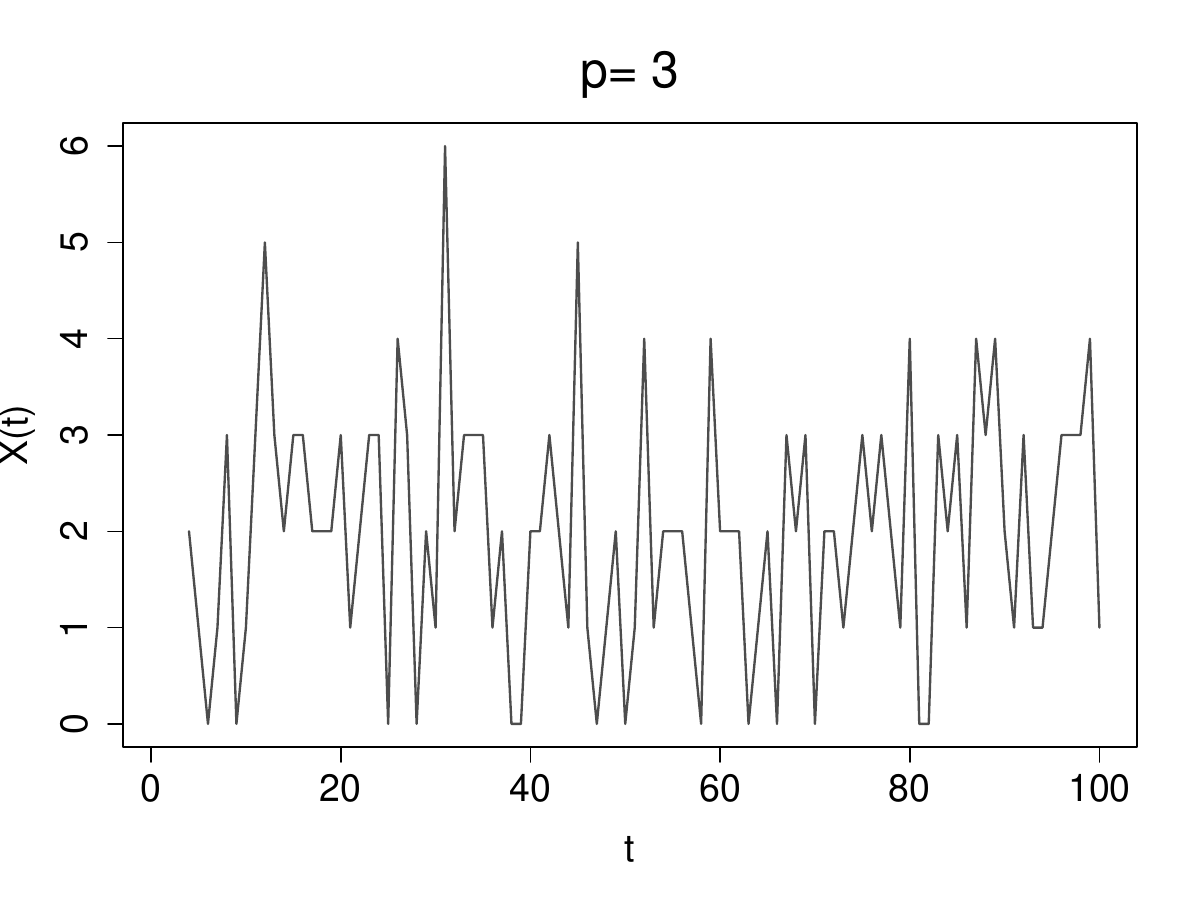}
\includegraphics[scale=0.27]{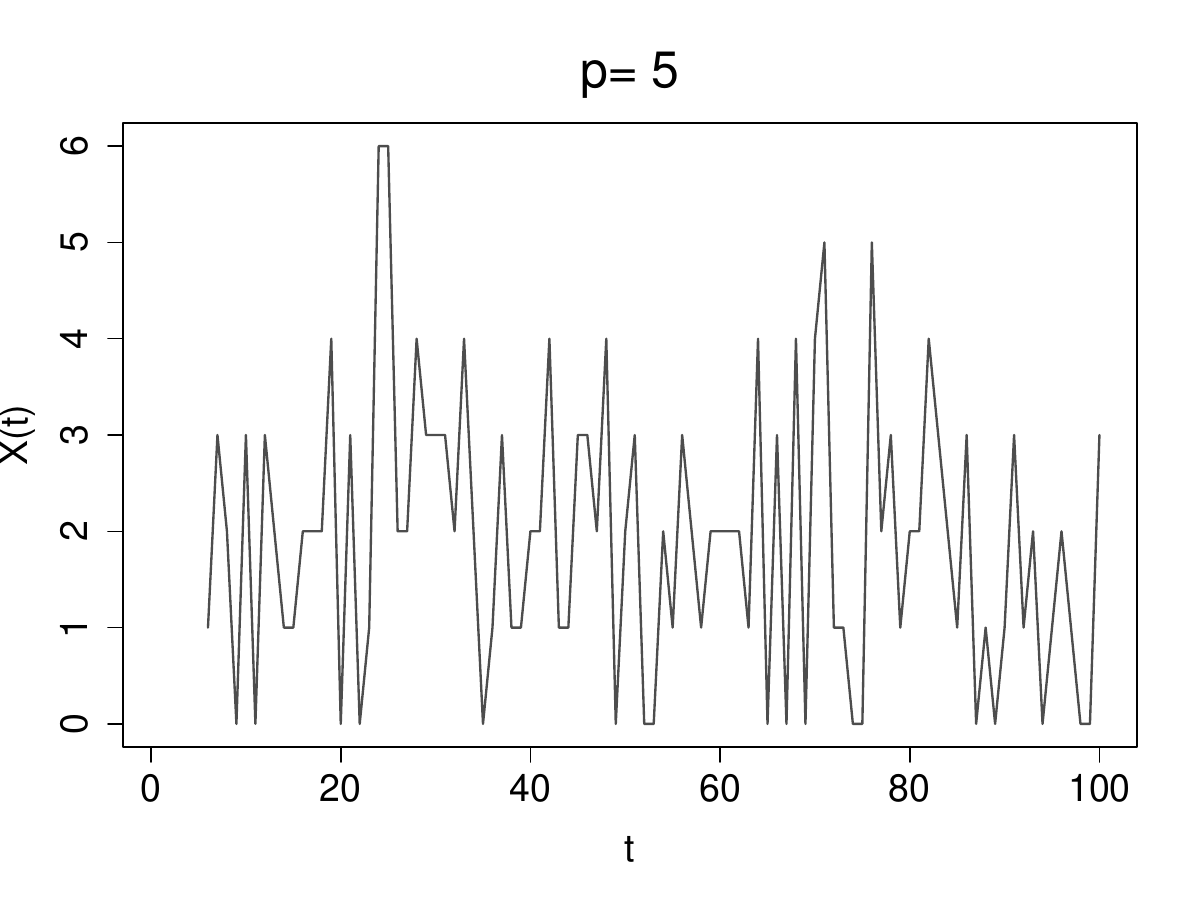}}
\vspace{-0mm}
\caption{{\small Simulated $\{X_t\}$ paths for $t=p+1,\ldots,100$. Across columns: $p=1$ (first), $p=3$ (middle) and $p=5$ (last). Across rows: type A (top) and type B (bottom).}}
\label{fig:sims}
\end{figure}

\begin{figure}
\centerline{\includegraphics[scale=0.8]{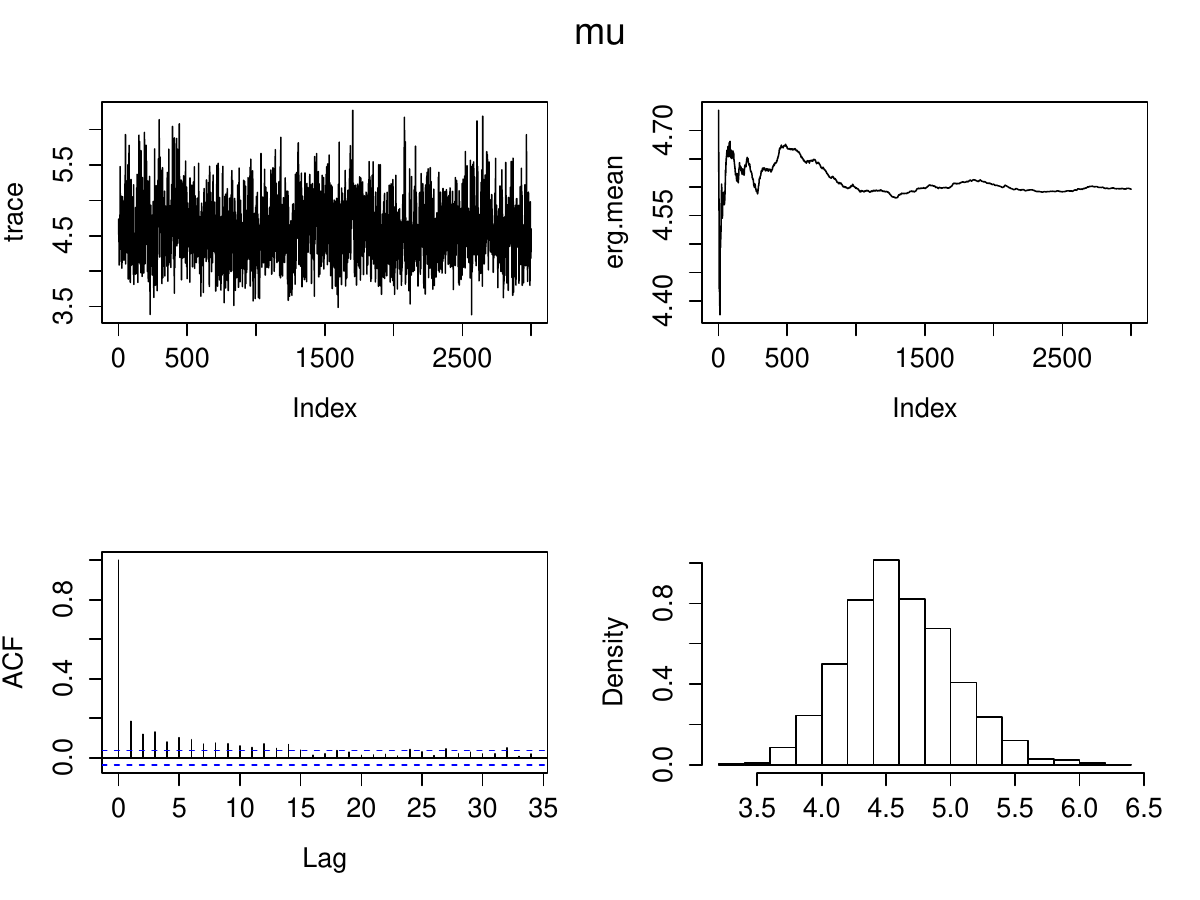}}
\vspace{-0mm}
\caption{{\small MCMC convergence diagnostics for $\mu$ in type A model for Coahuila state. Trace plot, ergodic means, autocorrelation function and probability histogram.}}
\label{fig:mcmc}
\end{figure}

\begin{figure}
\centerline{\includegraphics[scale=0.42]{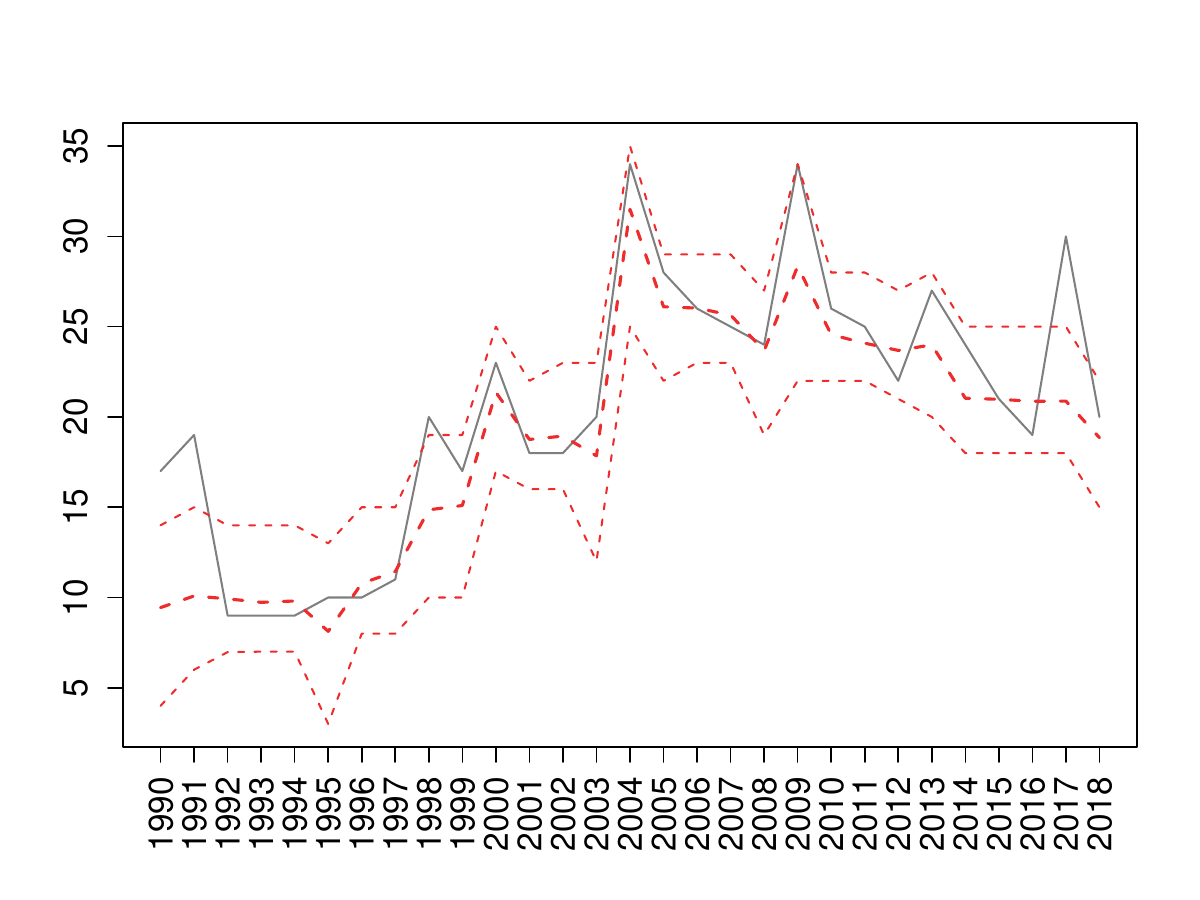}
\includegraphics[scale=0.42]{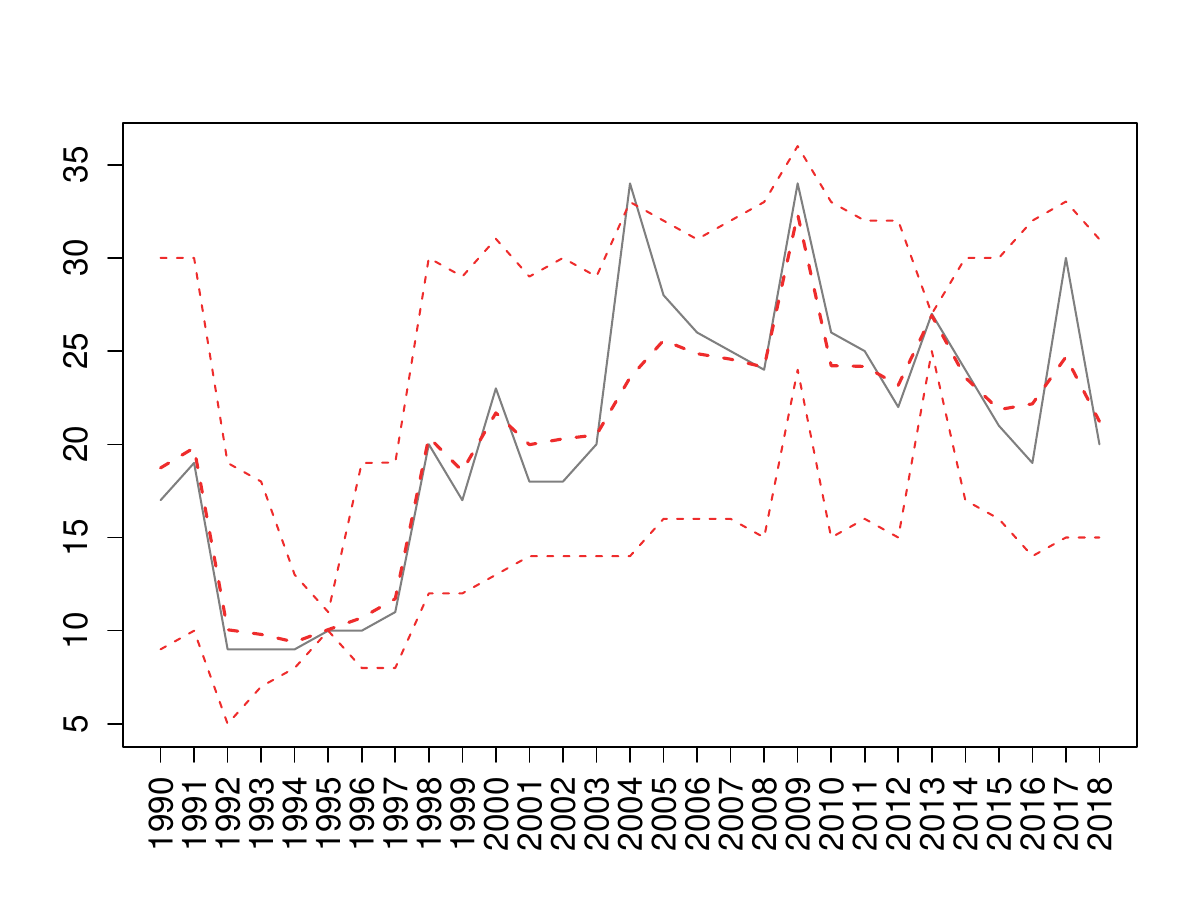}}
\vspace{-0mm}
\caption{{\small Maternal mortality for Baja California state. Best fitting models. Type A with $p=4$ (left) and type B with $p=5$ (right). Observed data (solid grey), point prediction (thick dotted red) and 95\% credible interval (dotted red).}}
\label{fig:baja}
\end{figure}

\begin{figure}
\centerline{\includegraphics[scale=0.42]{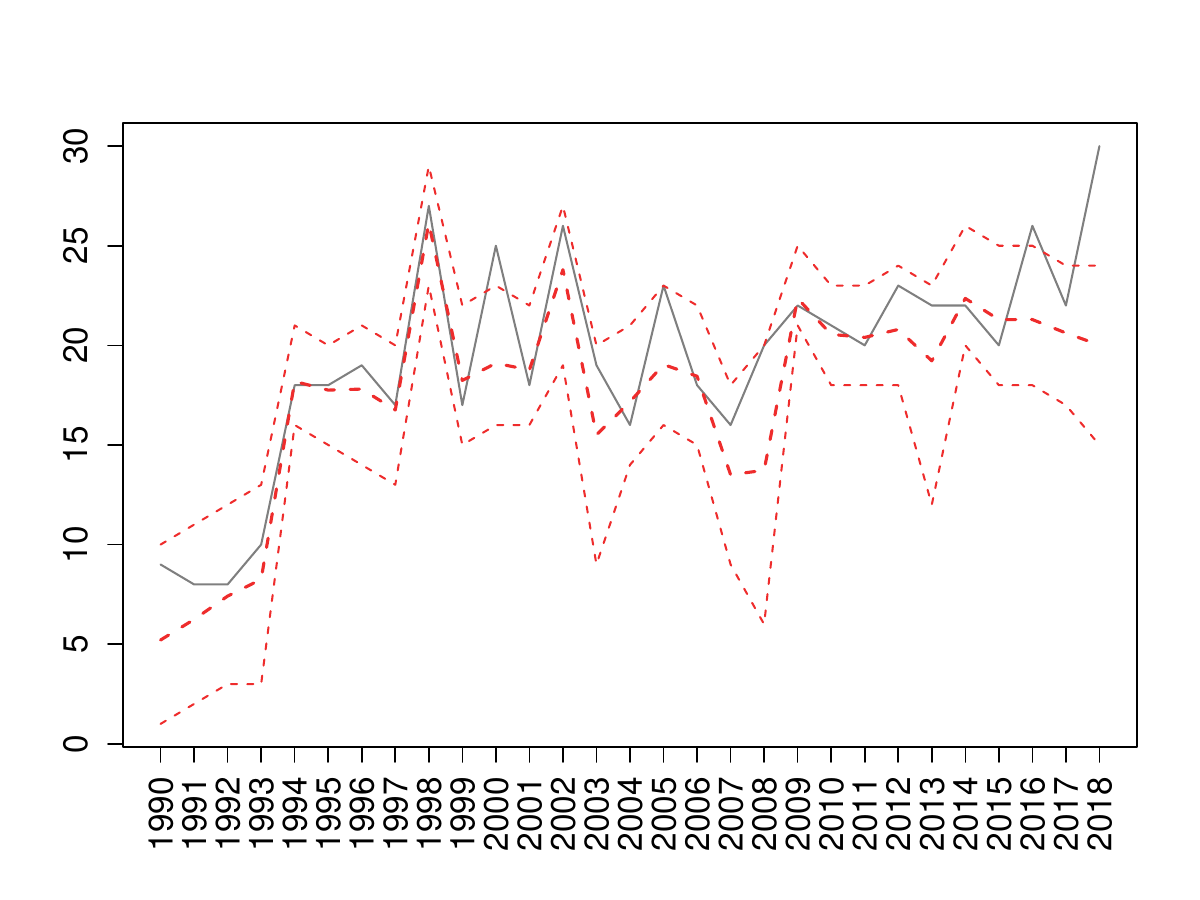}
\includegraphics[scale=0.42]{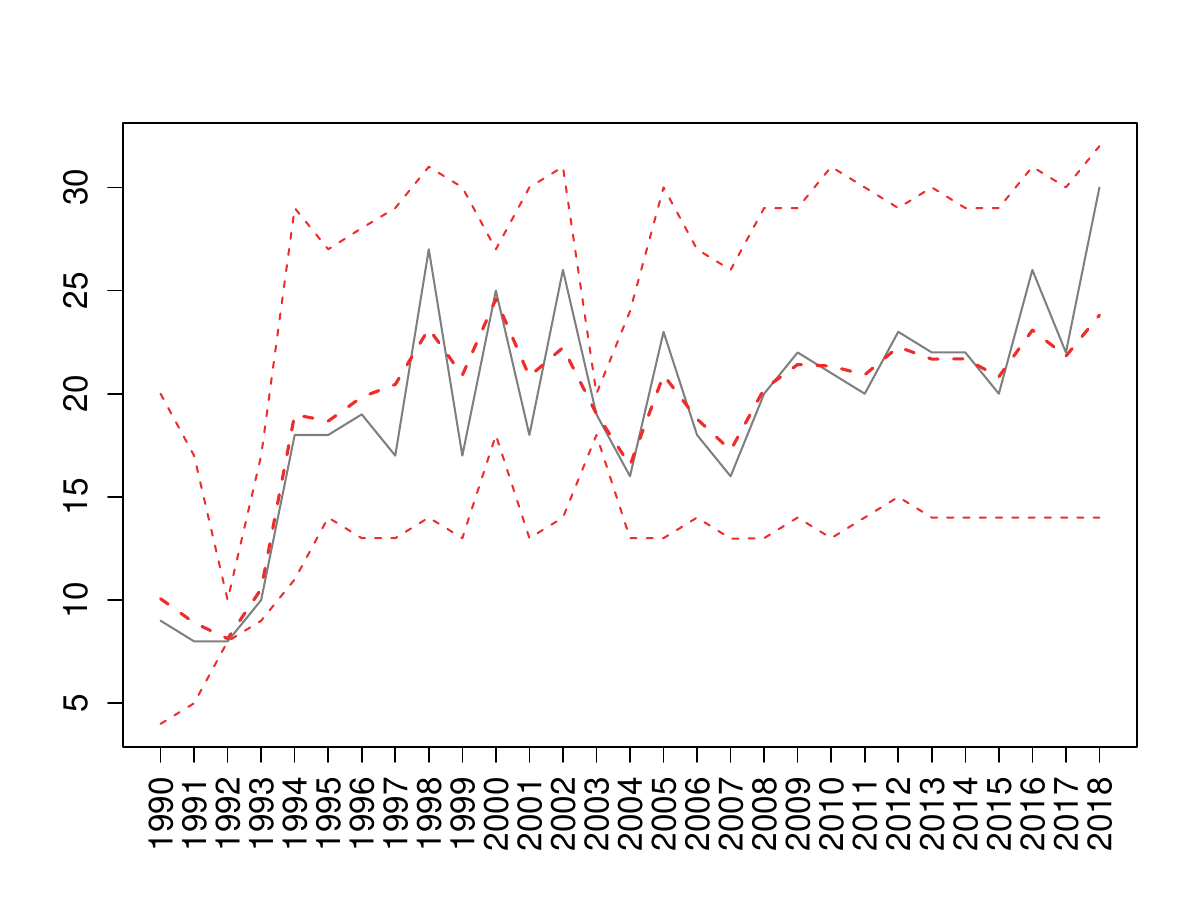}}
\vspace{-0mm}
\caption{{\small Maternal mortality for Coahuila state. Best fitting models. Type A with $p=4$ (left) and type B with $p=5$ (right). Observed data (solid grey), point prediction (thick dotted red) and 95\% credible interval (dotted red).}}
\label{fig:coahuila}
\end{figure}

\begin{figure}
\centerline{\includegraphics[scale=0.42]{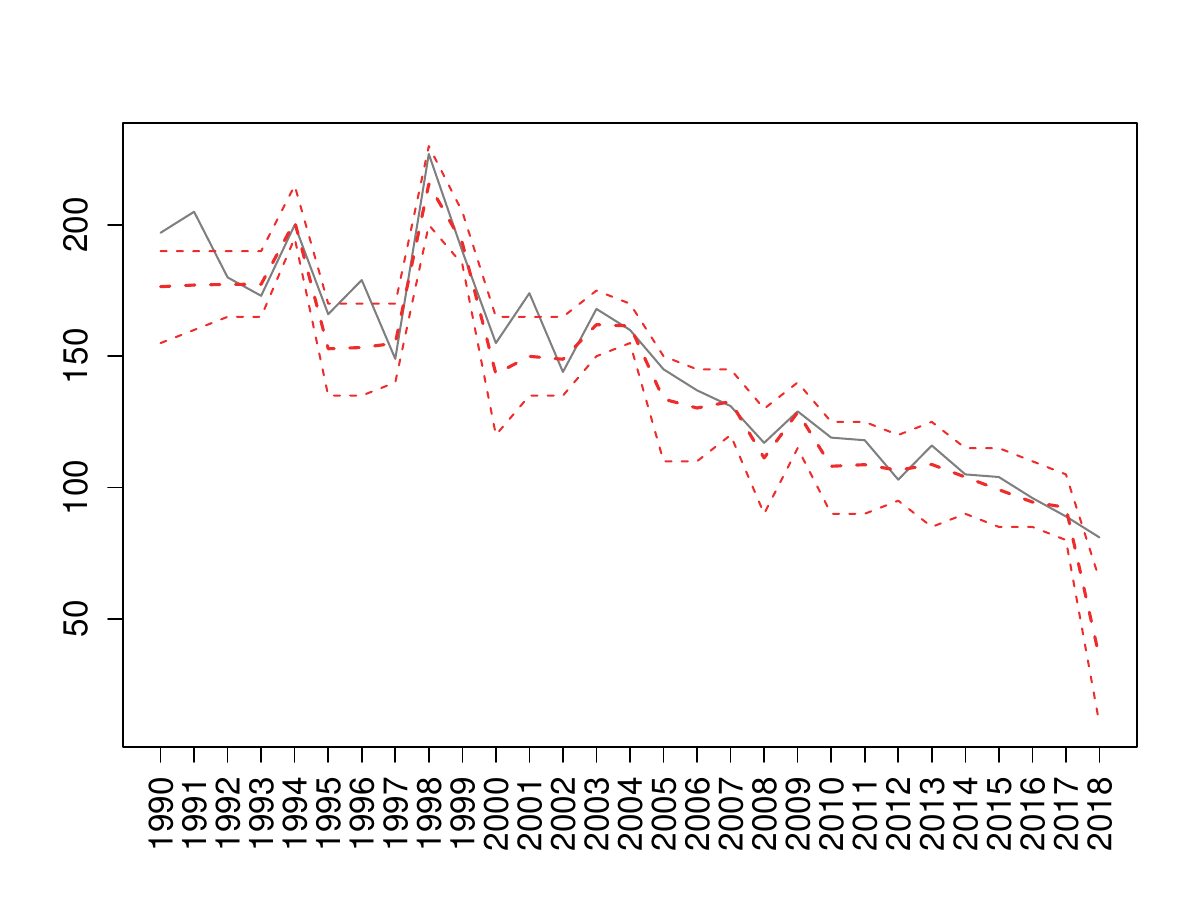}
\includegraphics[scale=0.42]{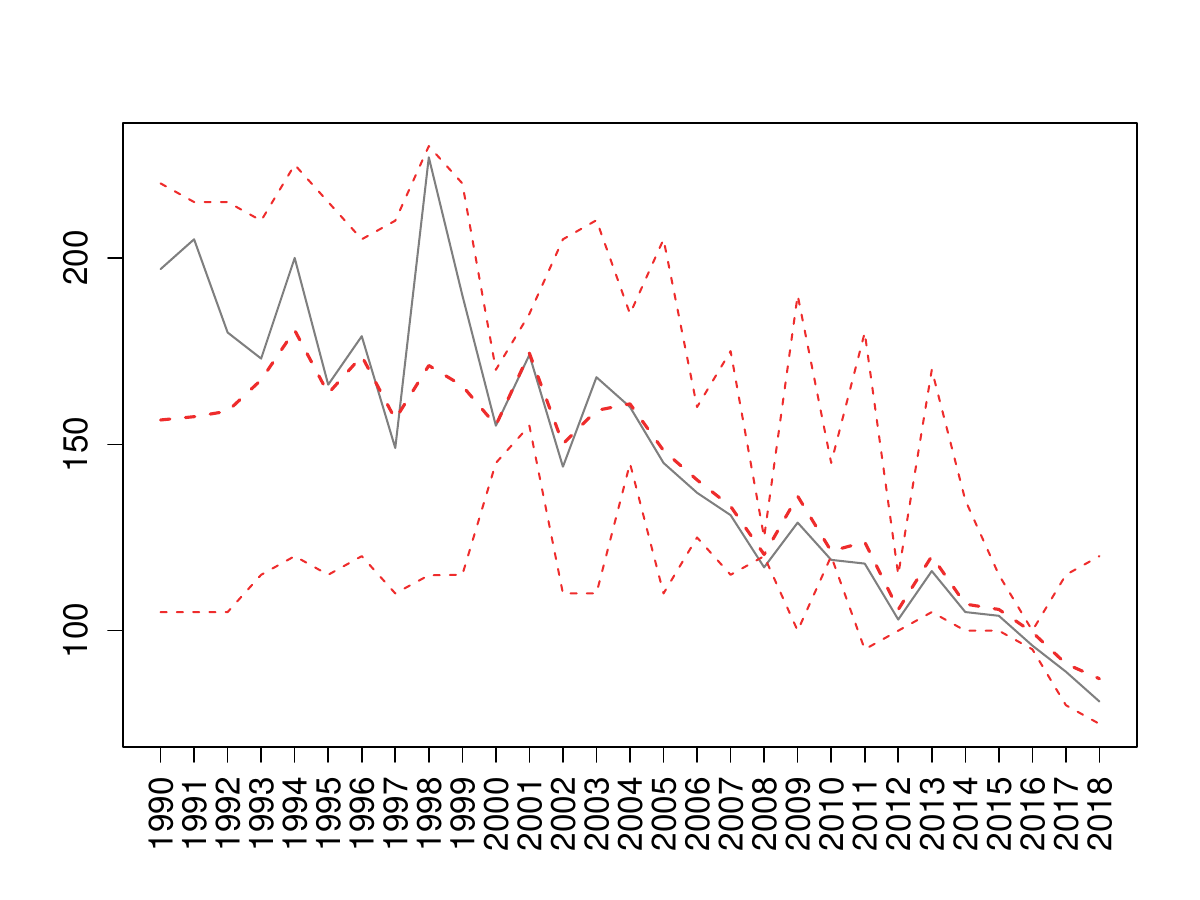}}
\vspace{-0mm}
\caption{{\small Maternal mortality for CDMX (Mexico City) state. Best fitting models. Type A with $p=4$ (left) and type B with $p=3$ (right). Observed data (solid grey), point prediction (thick dotted red) and 95\% credible interval (dotted red).}}
\label{fig:cdmx}
\end{figure}

\begin{figure}
\centerline{\includegraphics[scale=0.42]{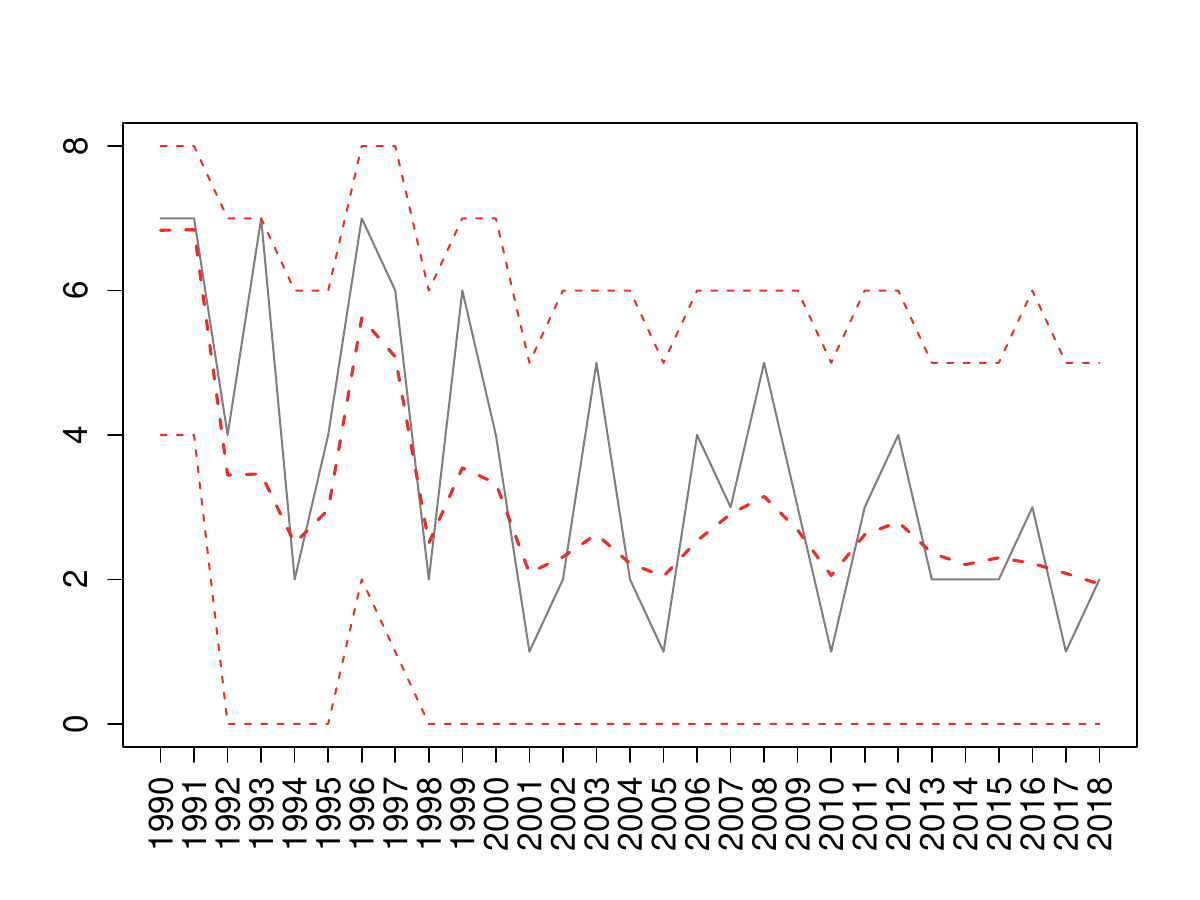}
\includegraphics[scale=0.42]{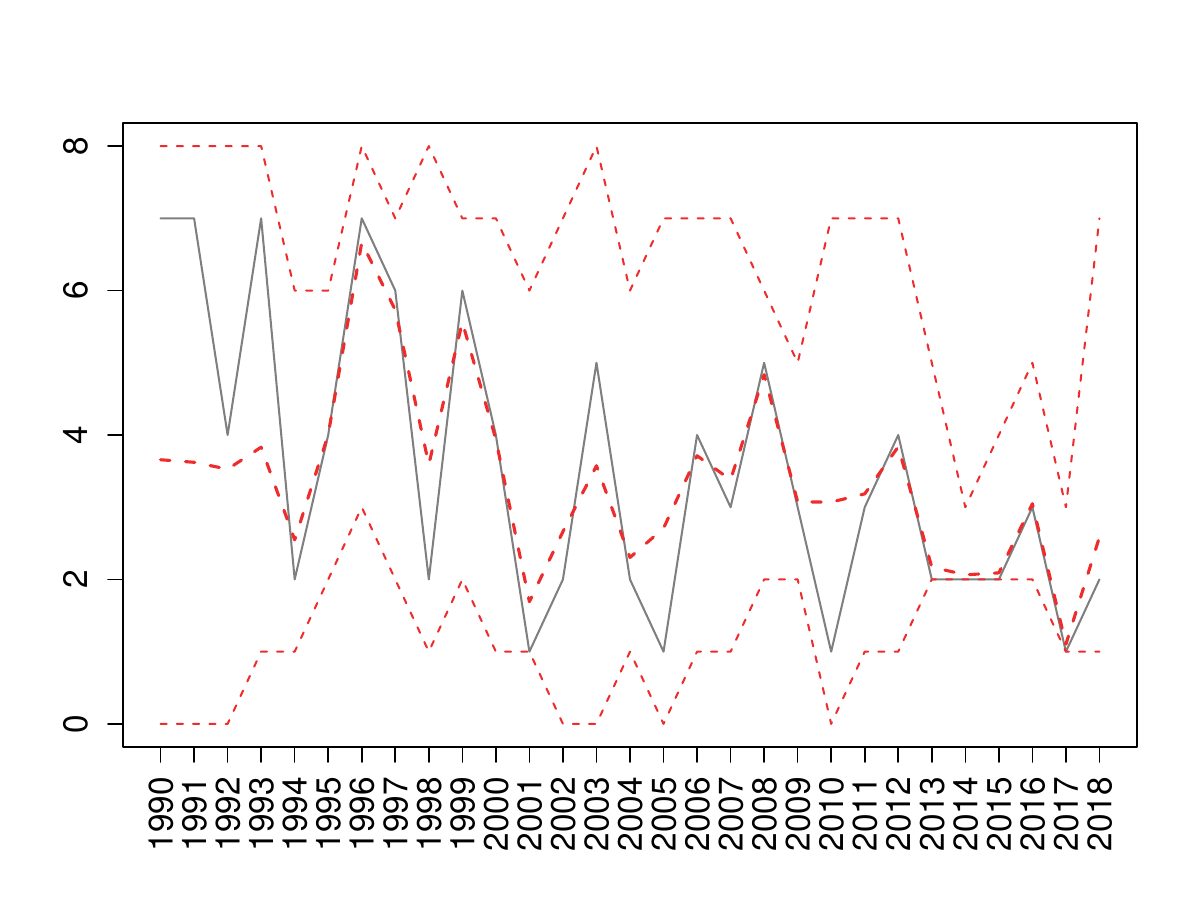}}
\vspace{-0mm}
\caption{{\small Maternal mortality for Colima state. Best fitting models. Type A with $p=1$ (left) and type B with $p=4$ (right). Observed data (solid grey), point prediction (thick dotted red) and 95\% credible interval (dotted red).}}
\label{fig:colima}
\end{figure}

\begin{figure}
\centerline{\includegraphics[scale=0.42]{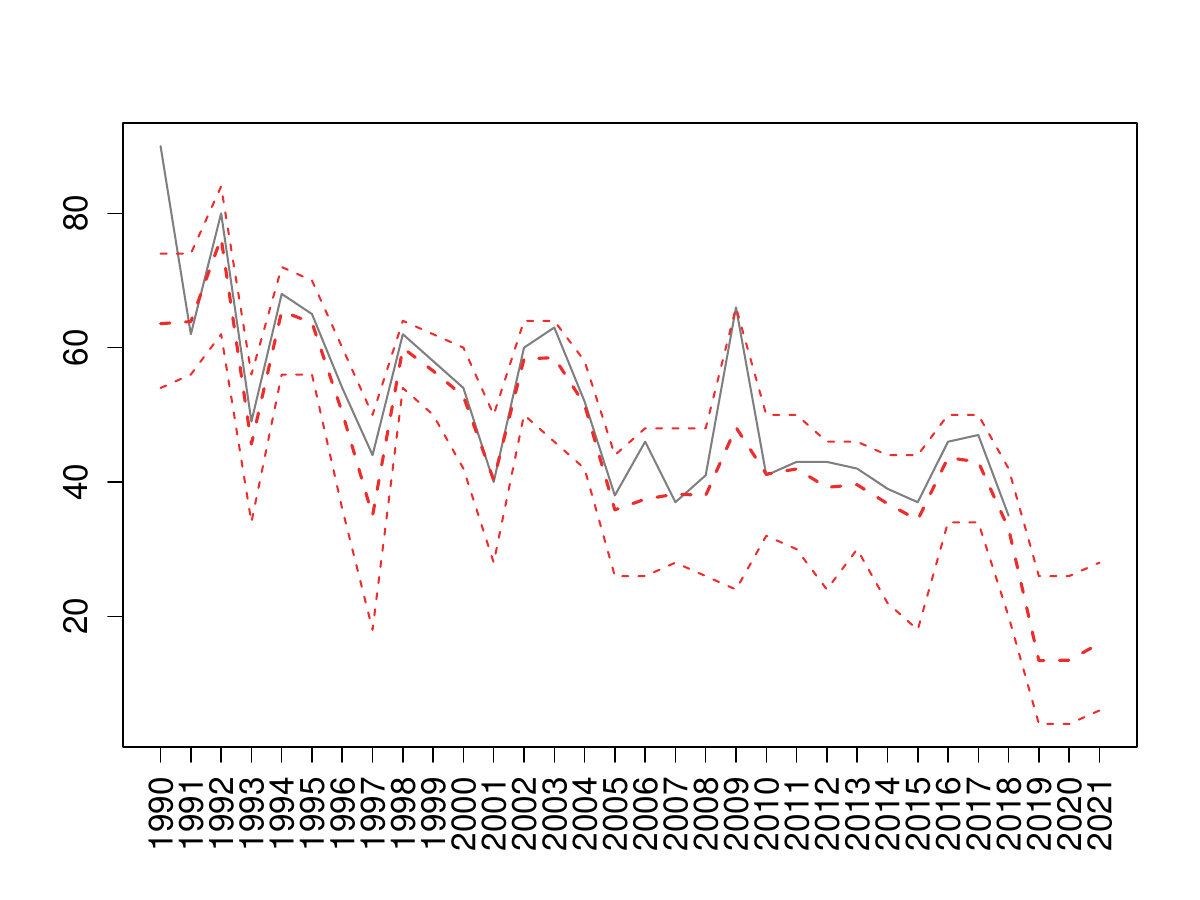}
\includegraphics[scale=0.42]{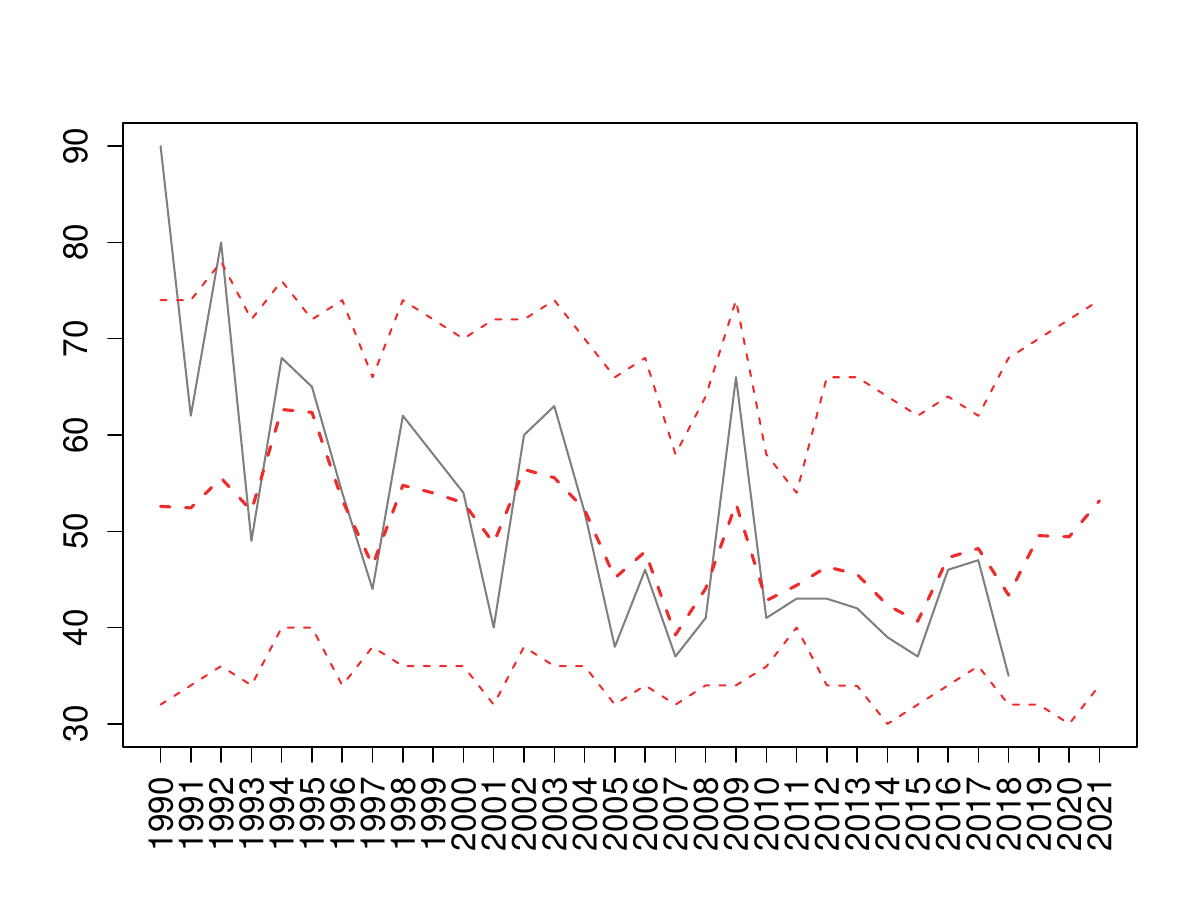}}
\centerline{\includegraphics[scale=0.42]{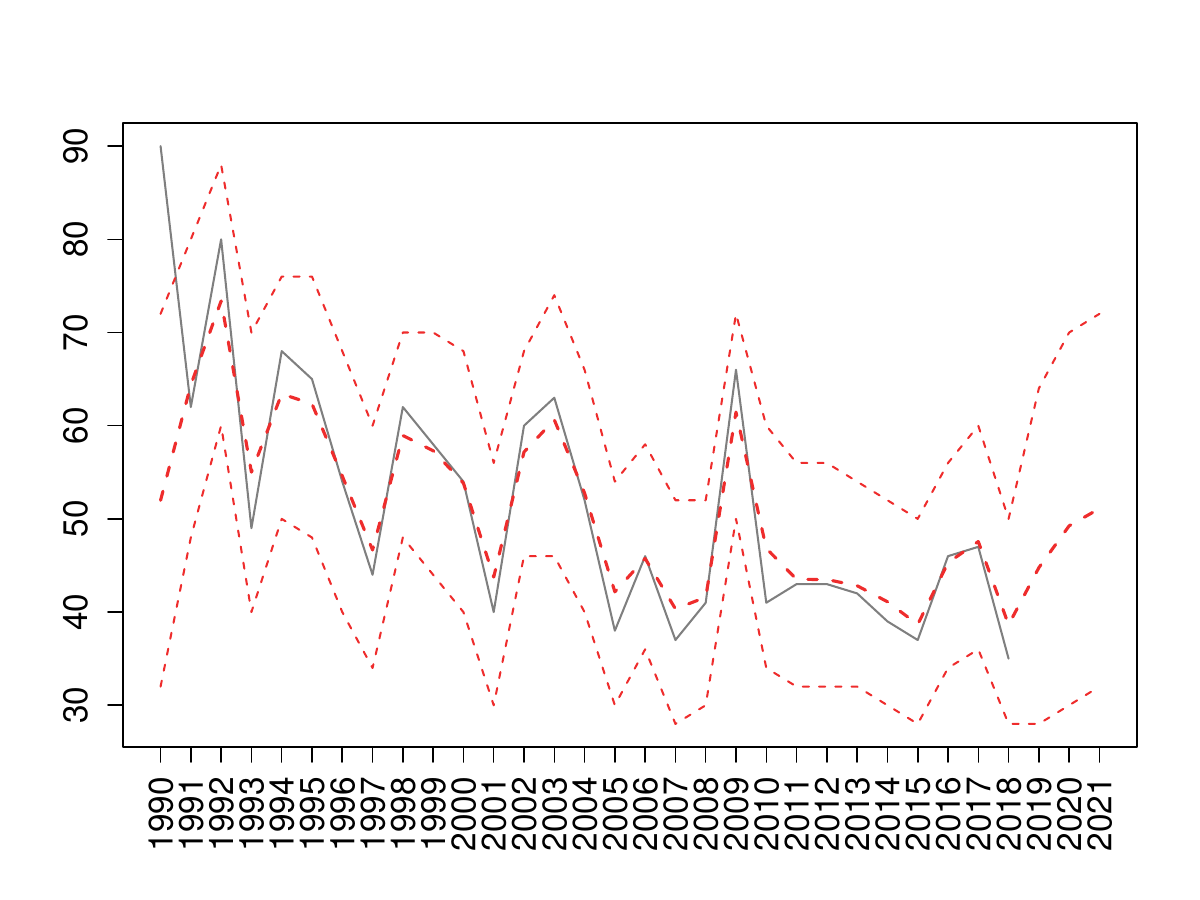}
\includegraphics[scale=0.42]{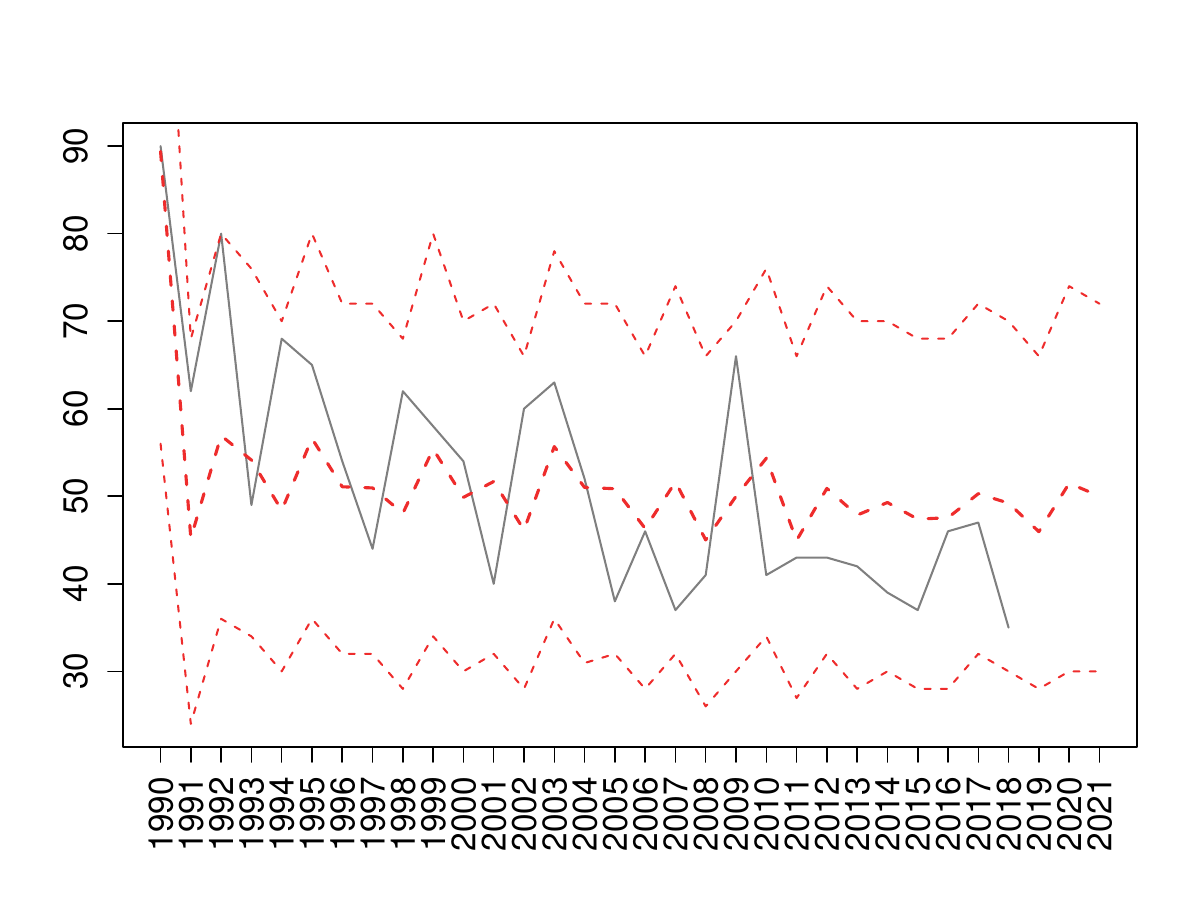}}
\vspace{-0mm}
\caption{{\small Maternal mortality for Guanajuato state. Best fitting models. Type A with $p=2$ (top left), type B with $p=2$ (top right), INAR(1) (bottom left) and INGARCH(1,1) bottom right. Observed data (solid grey), point prediction (thick dotted red) and 95\% credible interval (dotted red). All panels contain out of sample predictions for 3 years ahead.}}
\label{fig:gua}
\end{figure}

\end{document}